\documentclass{llncs}

\usepackage[final]{microtype}
\usepackage{lmodern}
\usepackage{breakcites}
\usepackage{amsmath,amssymb,xcolor,graphicx,wrapfig,paralist}
\usepackage{algorithmicx,algorithm}
\usepackage[noend]{algpseudocode}
\usepackage{booktabs}
\usepackage{proof}
\usepackage{tikz}
\usepackage{adjustbox}
\usepackage{csquotes}
\usepackage{graphics}
\usepackage{stmaryrd}
\usepackage[caption=false]{subfig}
\usepackage{thm-restate}
\usepackage{makecell}

\usepackage{xcolor}

\usepackage{pgfplots}
\usepackage{pgfplotstable}

\usepackage{ellipsis, mparhack, ragged2e} 
\usepackage[l2tabu, orthodox]{nag} 

\usetikzlibrary{shapes,positioning,arrows,calc,automata,matrix,fit,arrows.meta}
\tikzstyle{state}+=[minimum size = 6mm, inner sep=0,outer sep=1]
\tikzset{->,>=stealth'}

\usepackage{tabu,multirow}
\newcolumntype{L}{X}
\newcolumntype{R}{>{\raggedleft\arraybackslash}X}
\newcolumntype{C}{>{\centering\arraybackslash}X}

\input{macros}

\definecolor{color1}{HTML}{de2d26} 
\definecolor{graph1}{HTML}{e66101}
\definecolor{graph2}{HTML}{fdb863}
\definecolor{graph3}{HTML}{b2abd2}
\definecolor{graph4}{HTML}{5e3c99}

\newcommand{\secspace}{\vspace*{-0.5em}}
\newcommand{\ssecspace}{\vspace*{-0.2em}}
\newcommand{\myspace}{\vspace*{-0.3em}}

\title{Continuous-Time Markov Decisions based on Partial Exploration}

\author{Pranav Ashok\inst{1}\and Yuliya Butkova\inst{2}\and Holger Hermanns\inst{2}\and Jan K\v{r}et\'{i}nsk\'{y}\inst{1}}
\institute{Technical University of Munich, Germany\\
  \and
  Saarland University, Saarbr\"ucken, Germany}

\begin{document}

\maketitle

\secspace

\begin{abstract}
We provide a framework for speeding up algorithms for time-bounded reachability analysis of continuous-time Markov decision processes.
The principle is to find a small, but almost equivalent \emph{subsystem} of the original system and only analyse the subsystem.
Candidates for the subsystem are identified through simulations and iteratively enlarged until runs are represented in the subsystem with high enough probability.
The framework is thus dual to that of abstraction refinement.
We instantiate the framework in several ways with several traditional algorithms and experimentally confirm orders-of-magnitude speed ups in many cases.


\end{abstract} 

\secspace\ssecspace

\section{Introduction}
\ssecspace

\emph{Continuous-time Markov decision processes} (CTMDP)~\cite{Bertsekas95,Sennott99,Feinberg04} are the natural real-time extension of (discrete-time) Markov decision processes (MDP). They can likewise be viewed as non-deterministic extensions of continuous-time Markov chains (CTMC). As such, CTMDP feature probabilistic and non-deterministic behaviour as well as random time delays governed by exponential probability distributions. Prominent application areas of CTMDP include operations research~\cite{DBLP:journals/jacm/BrunoDF81,Feinberg04}, power management and scheduling~\cite{DBLP:journals/tcad/QiuQP01}, networked, distributed systems~\cite{DBLP:conf/srds/HaverkortHK00,DBLP:conf/sosp/GhemawatGL03}, as well as epidemic and population processes~\cite{doi:10.1287/opre.29.5.971}. Moreover, CTMDPs are the core semantic model underlying formalisms such as generalised stochastic Petri nets, Markovian stochastic activity networks, and interactive Markov chains~\cite{DBLP:conf/apn/EisentrautHK013}.

A large variety of properties  can be expressed using logics such as CSL ~\cite{DBLP:conf/cav/AzizSSB96}. 
Apart from classical techniques from the MDP context, the analysis of such properties relies fundamentally on the problem of time-bounded reachability (TBR), i.e.~\emph{what is the maximal/minimal probability to reach a given state within a given time bound}. Since this is the cornerstone of the analysis, a manifold of algorithms have been proposed for TBR  \cite{DBLP:conf/tacas/BaierHHK04,DBLP:conf/fsttcs/BrazdilFKKK09,DBLP:conf/qest/NeuhausserZ10,fearnley_et_al:LIPIcs:2011:3354,DBLP:journals/cor/BuchholzS11,DBLP:conf/fsen/HatefiH13,DBLP:conf/atva/ButkovaHHK15}.
While the algorithmic approaches are diverse, relying on uniformisation and various forms of discretization, they are mostly back-propagating the values computed, i.e.~in the form of value iteration.

Not surprisingly, all these algorithms naturally process the state space of the CTMDP in its entirety. In this work we instead suggest a framework that enables TBR analysis with guaranteed precision while often exploring only a small, property-dependent part of the state space. Similar ideas have appeared for (discrete-time) MDPs and unbounded reachability \cite{atva} or mean payoff \cite{cav17}. 
These techniques are based on \emph{asynchronous} value-iteration approaches, originally proposed in the probabilistic planning world, such as bounded real-time dynamic programming (BRTDP)~\cite{brtdp}. 
Intuitively, the back-propagation of values (value iteration steps) are not performed on all states in each iteration (synchronously), but always only the ``interesting'' ones are considered (asynchronously); in order to bound the error in this approach, one needs to compute both an under- and an over-approximation of the actual value.

In other words, the main idea is to keep track of (under- and over-)approximation of the value when accepting that we have no information about the values attained in certain states.
Yet if we can determine that these states are reached with very low probability, their effect on the actual value is provably negligible and thus the lack of knowledge only slightly increases the difference between the under- and over-approximations.
To achieve this effect, the algorithm of~\cite{atva} alternates between two steps: (i) simulating a run of the MDP using a (hopefully good) scheduler, and (ii) performing the standard value iteration steps on the states visited by this run.

It turns out that this idea cannot be transferred to the continuous-time setting easily.
In technical terms, the main issue is that the value iteration in this context takes the form of synchronous back-propagation, which when implemented in an asynchronous fashion results in memory requirements that tend to dominate the memory savings expectable due to partial exploration.

Therefore, we twist the above approach and present a yet simpler algorithmic strategy in this paper. Namely, our approach alternates between several simulation steps, and a subsequent run of TBR analysis only focussed on the already explored subsystem, instead of the entire state space. If the distance between under- and over-approximating values is small enough, we can terminate; otherwise, running more simulations extends the considered state subspace, thereby improving the precision in the next round. Each run of the TBR analysis provides a scheduler for the subsystem that can be extended to be a scheduler on the original model. The extended scheduler obtained upon termination of our algorithm is guaranteed to be optimal for the TBR problem on the original CTMDP (up to user-defined precision).

There are thus two largely independent components to the framework, namely (i) a heuristic how to explore the system via simulation, and (ii) an algorithm to solve time-bounded reachability on CTMDP. The latter is here instantiated with some of the classic algorithms mentioned above, namely the first discretization-based algorithm \cite{DBLP:conf/qest/NeuhausserZ10} and the two most competitive improvements over it \cite{DBLP:journals/cor/BuchholzS11,DBLP:conf/atva/ButkovaHHK15}, based on uniformisation and untimed analysis. The former basically boils down to constructing a scheduler resolving the non-determinism effectively. We instantiate this exploration heuristics in two ways. Firstly, we consider a scheduler returned by the most recent run of the respective TBR algorithm, assuming this to yield a close-to-optimal scheduler, so as to visit the most important parts of the state space, relative to the property in question. Secondly, since this scheduler may not be available when working with TBR algorithms that return only the value, we also employ a scheduler resolving choices uniformly. Although the latter may look very straightforward, it turns out to already speed up the original algorithm considerably in many cases. This is rooted in the fact that that scheduler best represents the available knowledge, since the uniform distribution is the one with maximal entropy.

Depending on the model and the property under study, different ratios of the state space entirety need to be explored to achieve the desired precision. Furthermore, our approach is able to exploit that the reachability objective is of certain forms, in stark contrast to the classic algorithm that needs to perform the same computation irrespective of the concrete set of target states. Still, the approach we propose will naturally profit from future improvements in effectiveness of classic TBR analysis.  

We summarize our contribution as follows:\ssecspace
\begin{itemize}
	\item We introduce a framework to speed up TBR algorithms for CTMDP and instantiate it in several ways.
	It is based on a partial, simulation-based exploration of the state space spanned by a model.
	\item We demonstrate its effectiveness in combination with several classic algorithms, obtaining orders of magnitude speed ups in many experiments. We also illustrate the limitations of this approach on cases where the state space needs to be explored almost in its entirety.
	\item We conclude that our framework is a generic add-on to arbitrary TBR algorithms, often saving considerably more work than introduced by its overhead.
\end{itemize}

\secspace

\section{Preliminaries} \label{sec:prelim}
\ssecspace

In this section, we introduce some central notions.

A \emph{probability distribution} on a finite set $X$ is a mapping $\rho: X \to [0,1]$, such that $\sum_{x\in X} \rho(x) = 1$.
$\Distributions(X)$ denotes the set of all probability distributions on $X$.

\label{sec:ctmcp}
\begin{definition}
  A \emph{continuous-time Markov decision process} (CTMDP) is a tuple $\Ctmdp=(\initstate, S, \Act, \bfR, \goals)$ where $S$ is a finite set of \emph{states}, $\initstate$ is the initial state, $\Act$ is a finite set of \emph{actions}, $\bfR: S \times \Act \times S \to \mathbb{R}_{\geq 0}$ is a rate matrix and $\goals \subseteq S$ is a set of \emph{goal} states.
\end{definition}

For a state $s \in S$ we define the set of \emph{enabled actions} $\Act(s)$ as follows: $\Act(s) = \{\alpha \in \Act ~|~ \bfR(s,\alpha,s') > 0 \text{ for some } s'\} $. Those states $s'$ for which $\bfR(s,\alpha,s') > 0$ form the set of \emph{successor states of $s$ via $\alpha$} which is denoted as $\Succ(s, \alpha)$. W.\,l.\,o.\,g. we require that all sets $\Act(s)$ and $\Succ(s,\alpha)$ are non-empty. A state $s$, s.\,t. $\forall \alpha \in \Act(s):~\Succ(s,\alpha) = \{s\}$ is called \emph{absorbing}.


For a given state $s$ and action $\alpha \in\Act(s)$, we denote by $\exit{s,\alpha} = \sum_{s'} \bfR(s,\alpha,s')$ the \emph{exit rate} of $\alpha$ in $s$ and
$\trans(s,\alpha,s') = \bfR(s,\alpha,s') / \exit{s,\alpha}$.

\begin{wrapfigure}[18]{r}{0.4\textwidth}
	\vspace{-1.5em}
	\centering
	\scalebox{0.75}{	\centering
	\subfloat[]{
		\begin{tikzpicture}[node distance = {0.5cm and 1.2cm}, 
		v/.style = {draw, circle, thick}, 
		relv/.style = {circle, thick}, 
		none/.style = {draw=none,fill=none,inner sep=0,outer sep=0}, 
		f/.style = {draw, thick, circle, accepting}, 
		edge/.style = {-Latex, thick}, 
		rel/.style = {}, 
		ris/.style = {}, 
		ev/.style = {font=\scriptsize}, 
		initial distance=0.5cm, every initial by arrow/.style={-Latex,thick}, initial text={}
		]
		\node (init) [v, initial above] {0};
		\node (s) [v, below = of init] {1};
		
		\node (l1) [v] at ($(init) + (0,-2.3)$) {2};
		\node (l2) [v, below = of l1] {3};
		\node (l3) [v, below = of l2] {4};
		\node (l4) [v, below = of l3] {5};
		\node (lc1) [] at ($(init) + (-1,-2)$) {};
		\node (rc1) [] at ($(init) + (1,-2)$) {};
		\node (lc4) [] at ($(lc1) + (0,-4)$) {};		
		\node (rc4) [] at ($(rc1) + (0,-4)$) {};		
		

		\node (g) [f]  at ($(s) + (0,-5.7)$)  {G};
		
		\node (xa) [none] at ($(s)!0.5!(rc1)$) {};
		\node (xb) [none] at ($(s)!0.5!(lc1)$) {};

		\path 
			(init)  edge [midway, right] node[] {$1.1$} (s)
			(s)	edge [midway, right, -, dashed, rel] node[] {$\alpha$} (xa)
					edge [midway, left, -, dashed, ris] node {$\beta$} (xb)	
			(l1)	edge [midway, left] node[ev] {$1$} (l2)
			(l2)	edge [midway, left] node[ev] {$1$} (l3)
			(l3)	edge [midway, left] node[ev] {$1$} (l4)
			(l4)	edge [midway, left] node[ev] {$1$} (g)
			(xb)	edge [right] node[ev] {$2$} (l1)
			(g)	edge [midway, left, loop right] node[] {$1$} (g)
			
	;
		\draw [edge, rounded corners, ris] (xb) node[ev]at($(xb)!0.5!(lc1)+(-0.3,0.05)$){$1.1$} -- ($(lc1)$) -- ($(lc4)$) -- (g);
		\draw [edge, rounded corners, rel] (xa) node[ev]at($(xa)!0.5!(rc1)+(0.3,0.05)$){$0.5$} -- ($(rc1)$) -- ($(rc4)$) -- (g);
		\end{tikzpicture}
		\label{fig:ctmdp:a}
	}%
		\subfloat[]{
		\begin{tikzpicture}[node distance = {0.5cm and 1.2cm}, 
		v/.style = {draw, circle, thick}, 
		relv/.style = {circle, thick, fill=teal}, 
		none/.style = {draw=none,fill=none,inner sep=0,outer sep=0}, 
		f/.style = {draw, thick, circle, accepting}, 
		edge/.style = {-Latex, thick}, 
		rel/.style = {}, 
		ris/.style = {}, 
		ev/.style = {font=\scriptsize}, 
		initial distance=0.5cm, every initial by arrow/.style={-Latex,thick}, initial text={}
		]
		\node (init) [v, initial above] {0};
		\node (s) [v, below = of init] {1};
		
		\node (l1) [v] at ($(init) + (0,-2.3)$) {2};
		\node (lc1) [] at ($(init) + (-1,-2)$) {};
		\node (rc1) [] at ($(init) + (1,-2)$) {};
		\node (lc4) [] at ($(lc1) + (0,-4)$) {};		
		\node (rc4) [] at ($(rc1) + (0,-4)$) {};		
		

		\node (g) [f]  at ($(s) + (0,-5.7)$)  {G};
		
		\node (xa) [none] at ($(s)!0.5!(rc1)$) {};
		\node (xb) [none] at ($(s)!0.5!(lc1)$) {};

		\path 
			(init)  edge [midway, right] node[] {$1.1$} (s)
			(s)	edge [midway, right, -, dashed, rel] node[] {$\alpha$} (xa)
					edge [midway, left, -, dashed, ris] node {$\beta$} (xb)	
			(l1)	edge [midway, left, loop below] node[ev] {$1$} (l1)
			(xb)	edge [right] node[ev] {$2$} (l1)
			(g)	edge [midway, left, loop right] node[] {$1$} (g)
			
	;
		\draw [edge, rounded corners, ris] (xb) node[ev]at($(xb)!0.5!(lc1)+(-0.3,0.05)$){$1.1$} -- ($(lc1)$) -- ($(lc4)$) -- (g);
		\draw [edge, rounded corners, rel] (xa) node[ev]at($(xa)!0.5!(rc1)+(0.3,0.05)$){$0.5$} -- ($(rc1)$) -- ($(rc4)$) -- (g);
		\end{tikzpicture}
		\label{fig:ctmdp:b}
	}%
	}
	\caption{Example CTMDPs.}
	\label{fig:ctmdp-prelim}
\end{wrapfigure}
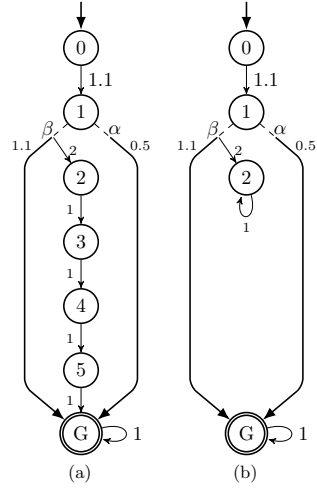

An example CTMDP is depicted in Fig. \ref{fig:ctmdp:a}. Here states are depicted in circles and are labelled with numbers from $0$ to $5$. The goal state $G$ is marked with a double circle. Dashed transitions represent available actions, e.g. state $1$ has two enabled actions $\alpha$ and $\beta$. A solid transition labelled with a number denotes the rate, e.g. $\bfR(1,\beta,G) = 1.1$, therefore there is a solid transition from state $1$ via action $\beta$ to state $G$ with rate $1.1$. If there is only one enabled action for a state, we only show the rates of the transition via this action and omit the action itself. For example, state $0$ has only 1 enabled action (lets say $\alpha$) and therefore it only has outgoing solid transition with rate $1.1 = \bfR(0,\alpha,1)$.


The system starts in the initial state $s_0 = \initstate$. While being in a state $s_0$, the system picks an action $\alpha_0 \in \Act(s)$. When an action is picked the CTMDP resides in $s_0$ for the amount of time $t_0$ which is sampled from exponential distribution with parameter $\exit{s_0,\alpha_0}$. Later in this paper we refer to this as \emph{residence time} in a state. After $t_0$ time units the system transitions into one of the successor states $s_1 \in \Succ(s_0, \alpha_0)$ selected randomly with distribution $\trans(s_0,\alpha_0,\cdot)$. After this transition the process is repeated from state $s_1$ forming an \emph{infinite path} $\path = s_0 \overset{\alpha_0,t_0}{\longrightarrow} s_1 \overset{\alpha_1,t_1}{\longrightarrow} s_2 \ldots$. A finite prefix of an infinite path is called a \emph{(finite) path}. We will use $\last{\path}$ to denote the last state of a finite path $\path$. We will denote the set of all finite paths in a CTMDP with $\paths$, and the set of all infinite paths with $\pathsinf$. 

CTMDPs pick actions with the help of \emph{schedulers}. A scheduler is a measurable\footnote{Measurable with respect to the standard $\sigma$-algebra on the set of paths~\cite{DBLP:conf/qest/NeuhausserZ10}.} function $\straa:\paths\times\Realsplus \to \Distributions(\Act)$ such that $\straa(\path, t) \in \Act(\last{\path})$. Being in a state $s$ at time point $t$ the CTMDP samples an action from $\straa(\path, t)$, where $\path$ is the path that the system took to arrive in $s$. We denote the set of all schedulers with $\straas$.

Fixing a scheduler $\straa$ in a CTMDP $\Ctmdp$, the unique probability measure $\bfP[\Ctmdp]{\straa}$ over the space of all infinite paths can be obtained ~\cite{DBLP:phd/de/Neuhausser2010}, denoted also by $\bfP[]{\straa}$ when $\Ctmdp$ is clear from context.

\myspace

\subsection*{Optimal Time-Bounded Reachability}
\ssecspace

Let $\Ctmdp = (\initstate, S, \Act, \bfR, \goals)$ be a CTMDP, $s \in S$, $T \in\Realsplus$ a time bound, and $\opt \in \{\sup, \inf\}$. The \emph{optimal (time-bounded) reachability probability} (or \emph{value}) of state $s$ in $\Ctmdp$ is defined as follows:
$$\val{\Ctmdp}{s}{T} := \opt\limits_{\straa \in \straas} \bfP[\Ctmdp]{\straa}\left[\treach{T} \goals\right],$$
where $\treach{T} \goals = \{s_0 \overset{\alpha_0,t_0}{\longrightarrow} s_1 \overset{\alpha_1,t_1}{\longrightarrow} s_2 \ldots \mid s_0 = s \land \exists i:  s_i \in \goals \land \sum_{j=0}^{i-1} t_j \leq T\}$ is the set of paths starting from $s$ and reaching $\goals$ before $T$. 

The \emph{optimal (time-bounded) reachability probability} (or \emph{value}) of $\Ctmdp$ is defined as $\val{\Ctmdp}{}{T} = \val{\Ctmdp}{\initstate}{T}$.
A scheduler that achieves optimum for $\val{\Ctmdp}{}{T}$ is the \emph{optimal scheduler}. A scheduler that achieves value $v$, such that $||v-\val{\Ctmdp}{}{T}||_{\infty}<\varepsilon$ is called \emph{$\varepsilon$-optimal}.


\secspace

\section{Algorithm}\label{sec:algo}

\ssecspace

\begin{wrapfigure}[11]{r}{0.5\textwidth}
	\vspace{-2em}
	\centering
	\includegraphics[scale=1]{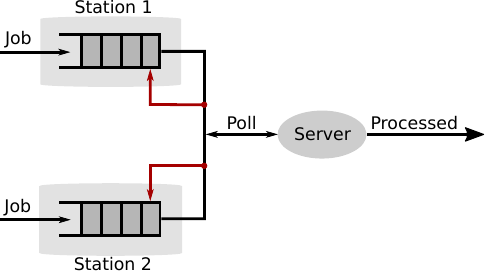}
	\caption{Schematic representation of polling system.}
	\label{fig:ps}
\end{wrapfigure}

In this work we target CTMDPs that have large state spaces, but only a small subset of those states is actually contributing significantly to the reachability probability. 

Consider, for example, the \emph{polling system} represented schematically in Figure \ref{fig:ps}. Here two stations store continuously arriving tasks in a queue. Tasks are to be processed by a server. If the task is processed successfully it is removed from the queue, otherwise it is returned back into the queue. State space of the CTMDP $\Ctmdp$ modelling this polling system is a tuple $(q_1,q_2,s)$, where $q_i$ is the amount of tasks in queue $i$ and $s$ is a state of the server (could be e.\,g. \emph{processing task, awaiting task}, etc.). 

One of the possible questions could be, for example, \emph{what is the maximum probability of both queues to be full after a certain time point}. This corresponds to goal states being of the form $(N,N,s)$, where $N$ is the maximal queue capacity and $s$ -- any state of the server. Given that both queues are initially empty, all the paths reaching goal states have to visit states $(q_1,q_2,\cdot)$, where $q_i = [0..N]$.

However, for similar questions, for example, \emph{what is the maximum probability of the first queue to be full after a certain time point}, the situation changes. Here goal states are of the form $(N,q_2,s)$, where $q_2 = 0..N$ and $s$ -- any state of the server. The scheduler that only extracts tasks from the second queue is the fastest to fill the first one and is therefore the optimal one. The set of states that are most likely visited when following this scheduler are those states where the size of the second queue is small. This naturally depends on the rates of task arrival and processing. Assuming that the size of the queue rarely exceeds 2 tasks, all the states $(\cdot,q_2,\cdot)$, where $q_2 = 3..N$ do not affect the reachability probability too much. 

As a more concrete example, consider the CTMDP of Fig. \ref{fig:ctmdp:a}. Here all the states in the centre have exit rate $1$ and form a long chain. Due to the length of this chain the probability to reach the goal state via these states within time 2 is very small. In fact, the maximum probability to reach the target state within 2 time units in the CTMDP on the left and the one on the right are exactly the same and equal $0.4584$. Thus, on this CTMDP, 40\%
of the state space can be reduced without any effect on the reachability value.

Classical model checking algorithms do not take into account any information about the property and perform exhaustive state-space exploration. Given that only a subset of states is relevant to the reachability value, these algorithms may perform many unnecessary computations. 

\ssecspace

\subsection*{Our Solution}\label{sec:algo:solution}
\ssecspace

Throughout this section we work with a CTMDP $\Ctmdp=(\initstate, S, \Act, \bfR, \goals)$ and a time bound $T \in \Realsplus$.

The main contribution of this paper is a simple framework for solving the time-bounded reachability objective in CTMDPs without considering their whole state-space. This framework in presented in Algorithm \ref{alg:main}.
The algorithm involves the following major steps:

\begin{algorithm}[h]
	\caption{$\Algo$}\label{alg:main}
	\hspace*{\algorithmicindent} \textbf{Input:}
	CTMDP $\Ctmdp=(\initstate, S, \Act, \bfR, \goals)$,
	time bound $T$,
	precision $\varepsilon$\\
	\hspace*{\algorithmicindent} \textbf{Output:} 
	$(\ell,u)\in [0,1]^2 
	$ such that $\ell \leqslant\val{}{}{T} \leqslant u$ and $u-\ell < \varepsilon$ and \\
	\hspace*{2cm}$\varepsilon-$ optimal scheduler $\straa$ for $\val{\Ctmdp}{}{T}$
	\begin{algorithmic}[1]
		
		\If{$\initstate \in G$}
			\Return $(1,1)$, and an arbitrary scheduler $\pi \in \straas$
		\EndIf
		
		\State $\ell=0,u=1$ 

		\State $\straasim = \straauni$
		\State $S' = \{\initstate\}$
		
		\While{$ u-\ell
			\geqslant \varepsilon$} \label{alg:main:condn}
		\State $S' = S' \cup \getRelevantSubset{\Ctmdp}{T}{\straasim}$ \label{alg:main:rel}
		\State $\underline{\Ctmdp} = \lowerSub{\Ctmdp}{S'}$,  
		$\overline{\Ctmdp} = \upperSub{\Ctmdp}{S'}$  \label{alg:main:lower-upper}
		\State	$\ell
					= \val{\underline{\Ctmdp}}{}{T}$, 
				$u
					= \val{\overline{\Ctmdp}}{}{T}$ \label{alg:main:lower-upper-bounds}
		\State $\straaopt \leftarrow$ optimal scheduler for $\val{\overline{\Ctmdp}}{}{T}$,
				$\lstraaopt \leftarrow$ optimal scheduler for $\val{\underline{\Ctmdp}}{}{T}$ \label{alg:main:lower-upper-sch}
		\State\label{alg:main:straasim} $\straasim = \ChooseScheduler(\straauni,\straaopt)$  \hfill // choose a scheduler for simulations
		\EndWhile
		
		\State 	$\forall t \in [0,T], \forall s \in S': \straa(s,t) = \lstraaopt(s,t)$
		\State 	$\forall t \in [0,T], \forall s \in S\setminus S': \straa(s,t) \leftarrow$ any $\alpha \in \Act(s)$ \hfill // extend optimal scheduler to $S$
		
		\State \Return $(\ell,u), \straa$ 
	\end{algorithmic}
\end{algorithm}

\begin{description}
	\item[Step 1] A ``relevant subset'' of the state-space $S' \subseteq S$ is computed (line \ref{alg:main:rel}).

	\item[Step 2] Using this subset, CTMDPs $\underline{\Ctmdp}$ and $\overline{\Ctmdp}$ are constructed (line \ref{alg:main:lower-upper}). We define functions $\upperSub{\Ctmdp}{S'}$ and $\lowerSub{\Ctmdp}{S'}$ later in this section.

	\item[Step 3] The reachability values of $\underline{\Ctmdp}$ and $\overline{\Ctmdp}$ are under- and over-approximations of the reachability value $\val{\Ctmdp}{}{T}$. The values are computed in line \ref{alg:main:lower-upper-bounds} along with the optimal schedulers in line \ref{alg:main:lower-upper-sch}.
	
	\item[Step 4] At line \ref{alg:main:straasim} a scheduler $\straasim$ is selected that is used later for obtaining the relevant subset.
	
	\item[Step 5] If the two approximations are sufficiently close, i.e. $\val{\overline{\Ctmdp}}{}{T} - \val{\underline{\Ctmdp}}{}{T} < \varepsilon$, 
	$\left[\val{\underline{\Ctmdp}}{}{T}, \val{\overline{\Ctmdp}}{}{T}\right]$ is the interval in which the actual reachability value lies. The algorithm is stopped and this interval along with the $\varepsilon$-optimal scheduler are returned.
    If not, the algorithm repeats from line \ref{alg:main:rel}, growing the relevant subset in each iteration. 
\end{description}

In the following section, we elucidate the first four steps and discuss several instantiations and variations of this framework.

\ssecspace

\subsection{Step 1: Obtaining the Relevant Subset} \label{sec:rel-subset}

\ssecspace

The main challenge of the approach is to extract a relatively small \emph{representative} set $S' \subseteq S$, for which $\val{\overline{\Ctmdp}}{}{T}$ and $\val{\underline{\Ctmdp}}{}{T}$ are close to the value $\val{\Ctmdp}{}{T}$ of the original model, i.e. $\val{\overline{\Ctmdp}}{}{T} - \val{\underline{\Ctmdp}}{}{T} < \varepsilon$. 
If this is possible, then instead of computing the probability of reaching goal in $\Ctmdp$, we can compute the same in $\overline{\Ctmdp}$ and $\underline{\Ctmdp}$ to get an $\varepsilon$-width interval in which the actual value is guaranteed to lie. If the sizes of $\overline{\Ctmdp}$ and $\underline{\Ctmdp}$ are relatively small, then the computation is generally much faster.

In this work we propose a heuristics for selecting the relevant subset based on simulations. Simulation of continuous-time Markov chains (CTMDPs with singleton set $\Act(s)$ for all states) is a widely used approach that performs very well in many practical cases. It is based on sampling a path of the model 
according to its probability space.
Namely, upon entering a state $s$ the residence time is sampled from the exponential distribution and then the successor state $s'$ is sampled randomly from the distribution $\trans(s,\alpha,s')$. Here $\alpha$ is the only action available in state $s$. The process is repeated from state $s'$ until a goal state is reached or the cumulative time over this path exceeds the time-bound.

However this approach only works for fully stochastic processes, which is not the case for arbitrary CTMDPs due to the presence of multiple available actions. In order to make the process fully stochastic one has to fix a scheduler that decides which actions are to be selected during the run of a CTMDP. 

\begin{algorithm}[t]
 \caption{$\getRelevantSubset{\Ctmdp}{T}{\straasim}$}\label{alg:get-relevant}
 \hspace*{\algorithmicindent} \textbf{Input:}
 	CTMDP $\Ctmdp=(\initstate, S, \Act, \bfR, \goals)$,
	time bound $T$,
	a scheduler $\straasim$\\
 \hspace*{\algorithmicindent} \textbf{Parameters:} 
	$\nSim \in \Naturals$ \\
 \hspace*{\algorithmicindent} \textbf{Output:} 
    $S' \subseteq S$
\begin{algorithmic}[1]
	\For{$(i = 0;~ i<\nSim;~ i=i+1)$}
		\State $\path = \initstate$
		\State $t = 0$
		\While{$t < T$ \textbf{and} $\last{\path} \not \in G$} 
			\State $s = \last{\path}$
			\State Sample action $\alpha$ from distribution $\Distributions(\Act(s)) = \straasim(\path, 0)$
			\State Sample $t'$ from exponential distribution with parameter $\exit{s, \alpha}$
			\State Sample a successor $s'$ of $s$ with distribution $\trans(s, \alpha, \cdot)$
			\State $\path = \path \overset{t'}{\longrightarrow} s'$
			\State $t = t + t'$
		\EndWhile
	
		\State add all states of $\path$ to $S'$
	\EndFor

\end{algorithmic}
\end{algorithm}

%
%
%
%

\begin{figure}[t]
	\subfloat[]{
		\begin{tikzpicture}[node distance = {0.5cm and 1.2cm}, v/.style = {draw, thick, circle}, f/.style = {draw, thick, circle, accepting}, edge/.style = { -Latex}, dottededge/.style = {dotted, -Latex}, initial distance=0.5cm, every initial by arrow/.style={-Latex,thick}, initial text={}]
		\node (init) [v,initial above] {};
		
		\node (l1) [v, below left = of init] {};
		\node (l2) [v, below = of l1] {};
		\node (l3) [v, below = 1.4 of l2] {};
		\node (l4) [v, below = of l3] {};
		
		\node (r1) [v, below right = of init] {};
		\node (r2) [v, below = of r1] {};
		\node (r3) [v, below = 1.4 of r2] {};
		\node (r4) [v, below = of r3] {};
		
		\node (c1) [v] at ($(init) + (-0.6,-2)$) {};
		\node (c2) [v, below = of c1] {};
		\node (c3) [f, below = of c2] {};
		
		\draw [edge] (init) -> (l1);
		\draw [edge] (l1) -> (l2);
		\draw [dottededge] (l2) -> (l3);
		\draw [edge] (l3) -> (l4);
		
		\draw [edge] (init) -> (r1);
		\draw [edge] (r1) -> (r2);
		\draw [dottededge] (r2) -> (r3);
		\draw [edge] (r3) -> (r4);
		
		\draw [edge, rounded corners] (init) -- ($(init) + (-0.6,-0.5)$) -- (c1);
		\draw [edge] (c1) -> (c2);
		\draw [edge] (c2) -> (c3);
		\draw [edge,rounded corners] (c3) -- ($(c3) + (1.2,0.5)$) -- ($(init) + (0.6,-0.5)$) -- (init) ;
		\draw [edge,rounded corners] (l4) -- ($(l4) + (0.5,0)$) -- (c3) ;
		\end{tikzpicture}
, 		\label{fig:running-a}}%
	\hfill
	\subfloat[]{
		\begin{tikzpicture}[node distance = {0.5cm and 1.2cm}, v/.style = {draw, thick, circle, black!20}, f/.style = {draw, thick, circle, accepting, black!20}, edge/.style = {-Latex, black!20}, dottededge/.style = {dotted, -Latex}, initial distance=0.5cm, every initial by arrow/.style={-Latex,thick,highlight}, initial text={}, highlight/.style = {graph1, thick}]
		\node (init) [v, highlight,initial above] {};
		
		\node (l1) [v, highlight, below left = of init] {};
		\node (l2) [v, highlight, below = of l1] {};
		\node (l3) [v, below = 1.4 of l2] {};
		\node (l4) [v, below = of l3] {};
		
		\node (r1) [v, below right = of init] {};
		\node (r2) [v, below = of r1] {};
		\node (r3) [v, below = 1.4 of r2] {};
		\node (r4) [v, below = of r3] {};
		
		\node (c1) [v] at ($(init) + (-0.6,-2)$) {};
		\node (c2) [v, below = of c1] {};
		\node (c3) [f, below = of c2] {};
		
		\draw [edge, highlight] (init) -> (l1);
		\draw [edge, highlight] (l1) -> (l2);
		\draw [dottededge, black!10] (l2) -> (l3);
		\draw [edge] (l3) -> (l4);
		
		\draw [edge] (init) -> (r1);
		\draw [edge] (r1) -> (r2);
		\draw [dottededge, black!20] (r2) -> (r3);
		\draw [edge] (r3) -> (r4);
		
		\draw [edge, rounded corners] (init) -- ($(init) + (-0.6,-0.5)$) -- (c1);
		\draw [edge] (c1) -> (c2);
		\draw [edge] (c2) -> (c3);
		\draw [edge,rounded corners] (c3) -- ($(c3) + (1.2,0.5)$) -- ($(init) + (0.6,-0.5)$) -- (init) ;
		\draw [edge,rounded corners] (l4) -- ($(l4) + (0.5,0)$) -- (c3) ;
		\end{tikzpicture}
		\label{fig:running-b}}%
	\hfill
	\subfloat[]{
		\begin{tikzpicture}[node distance = {0.5cm and 1.2cm}, v/.style = {draw, circle, thick, black!20}, f/.style = {draw, thick, circle, accepting, black!20}, edge/.style = {-Latex, black!20}, dottededge/.style = {dotted, -Latex}, initial distance=0.5cm, every initial by arrow/.style={-Latex,thick,graph1}, initial text={}, highlight/.style = {graph1, thick}]
		\node (init) [v, highlight, initial above] {};
		
		\node (l1) [v, below left = of init] {};
		\node (l2) [v, below = of l1] {};
		\node (l3) [v, below = 1.4 of l2] {};
		\node (l4) [v, below = of l3] {};
		
		\node (r1) [v, below right = of init] {};
		\node (r2) [v, below = of r1] {};
		\node (r3) [v, below = 1.4 of r2] {};
		\node (r4) [v, below = of r3] {};
		
		\node (c1) [v, highlight] at ($(init) + (-0.6,-2)$) {};
		\node (c2) [v, highlight, below = of c1] {};
		\node (c3) [f, highlight, below = of c2] {};
		
		\draw [edge] (init) -> (l1);
		\draw [edge] (l1) -> (l2);
		\draw [dottededge, black!10] (l2) -> (l3);
		\draw [edge] (l3) -> (l4);
		
		\draw [edge] (init) -> (r1);
		\draw [edge] (r1) -> (r2);
		\draw [dottededge, black!20] (r2) -> (r3);
		\draw [edge] (r3) -> (r4);
		
		\draw [edge, highlight, rounded corners] (init) -- ($(init) + (-0.6,-0.5)$) -- (c1);
		\draw [edge, highlight] (c1) -> (c2);
		\draw [edge, highlight] (c2) -> (c3);
		\draw [edge,rounded corners] (c3) -- ($(c3) + (1.2,0.5)$) -- ($(init) + (0.6,-0.5)$) -- (init) ;
		\draw [edge,rounded corners] (l4) -- ($(l4) + (0.5,0)$) -- (c3) ;
		\end{tikzpicture}
		\label{fig:running-c}}%
	\caption{A simple CTMDP is presented in Fig. (\ref{fig:running-a}) with rates and action labels ignored. Fig. (\ref{fig:running-b}) shows a sampled run which ends on running out of time while exploring the left-most branch. Fig. (\ref{fig:running-c}) shows a simulation which ends on discovering a target state.}
\end{figure}
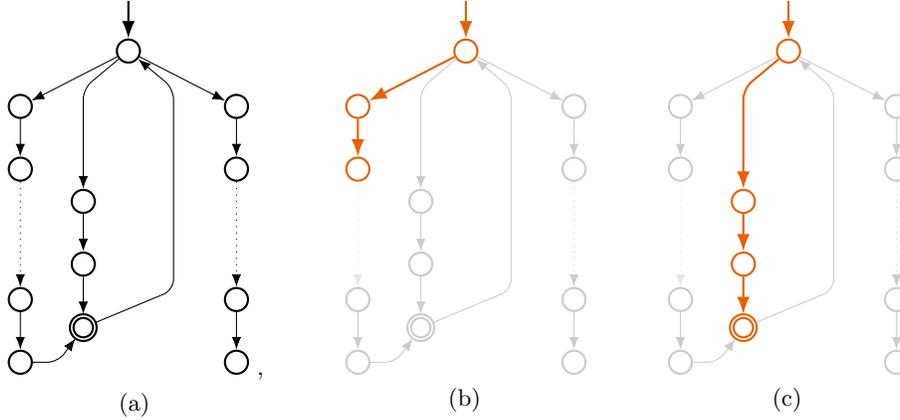

Our heuristic is presented in Algorithm \ref{alg:get-relevant}. It takes as input the CTMDP, time bound and a scheduler $\straasim$. 
The algorithms performs $\nSim$ simulations and outputs all the states visited during the execution.
Here $\nSim \in \Naturals$ is a parameter of the algorithm. 
Each simulation run starts in the initial state.
At first an action is sampled from $\Distributions(\Act(s)) = \straasim(\path,0)$ and then the simulation proceeds in the same way as described above for CTMCs by sampling residence times and successor states.
Notice that even though time-point $0$ is used for the scheduler, this does not affect the correctness of the approach, since it is only used as a heuristic to sample the subspace. 
In fact, one could instantiate $\getRelevantSubset{}{}{}$ with an arbitrary heuristic (e.\,g. from artificial intelligence domain, or one that is more targeted towards a specific model).
Correctness of the lower and upper bounds will not be affected by this.
However, termination of the algorithm cannot be ensured for any arbitrary heuristic.
Indeed, one has to make sure that the bounds will eventually converge to the value.

\begin{example}\label{epl:running}
Consider the CTMDP from Figure \ref{fig:running-a}. Figures \ref{fig:running-b} and \ref{fig:running-c} show two possible sampled paths.
The path in \ref{fig:running-c} reaches the target within the given time-bound and the path in \ref{fig:running-b} times out before reaching the goal state.
The relevant subset is thus all the states visited during the two simulations.
\end{example}

\ssecspace

\subsection{Step 2: Under- and Over-Approximating CTMDP}

\ssecspace

We will now explain line \ref{alg:main:lower-upper} of Algorithm \ref{alg:main}. Here we obtain two CTMDPs, such that the value of  $\underline{\Ctmdp}$ is a guaranteed lower bound, and the value of $\overline{\Ctmdp}$ is a guaranteed upper bound on the value of $\Ctmdp$.

Let $S' \subseteq S$ be the subset of states obtained in line \ref{alg:main:rel}. We are interested in extracting some information regarding the reachability value of $\Ctmdp$ from this subset. In order to do this, we consider two cases. (i) A pessimistic case, where all the unexplored states are non-goal states and absorbing (or \emph{sink states}); and (ii) an optimistic case, where all the unexplored states are indeed goals. It is easy to see that the ``pessimistic'' CTMDP $\underline{\Ctmdp}$ will have a smaller (or equal) value than the original CTMDP, which in turn will have a value smaller (or equal) than the ``optimistic'' CTMDP $\overline{\Ctmdp}$. Notice that for the reachability value the goal states can also be made absorbing and this will not change the value\footnote{This is due to the fact that for the reachability value, only what happens before the first arrival to the goal matters, and everything that happens afterwards is irrelevant.}. 
Before we define the two CTMDPs formally, we illustrate the construction on an example. Note that the fringe ``one-step outside'' of the relevant subset is still a part of the considered sub-CTMDPs.

\begin{figure}[t]
	\centering
	\subfloat[]{
		\begin{adjustbox}{width=0.24\textwidth}%
		\begin{tikzpicture}[node distance = {0.5cm and 1.2cm}, 
		v/.style = {draw, circle, thick, black!90}, 
		f/.style = {draw, thick, circle, accepting, black!90}, 
		dim/.style = {black!20},
		edge/.style = {-Latex, black!90}, 
		initial distance=0.5cm, every initial by arrow/.style={-Latex,thick, black!90}, initial text={}
		]
		\node (init) [v, initial above] {};
		
		\node (l1) [v, below left = of init] {};
		\node (l2) [v, below = of l1] {};
		\node (l3) [v, dim, below = 1.4 of l2] {};
		\node (l4) [v, dim, below = of l3] {};
		
		\node (r1) [v, dim, below right = of init] {};
		\node (r2) [v, dim, below = of r1] {};
		\node (r3) [v, dim, below = 1.4 of r2] {};
		\node (r4) [v, dim, below = of r3] {};
		
		\node (c1) [v] at ($(init) + (-0.6,-2)$) {};
		\node (c2) [v, below = of c1] {};
		\node (c3) [f, below = of c2] {};
		
		\draw [edge] (init) -> (l1);
		\draw [edge] (l1) -> (l2);
		\draw [dotted, thick, black!10] (l2) -> (l3);
		\draw [edge, dim] (l3) -> (l4);
		
		\draw [edge, dim] (init) -> (r1);
		\draw [edge, dim] (r1) -> (r2);
		\draw [dotted, thick, dim] (r2) -> (r3);
		\draw [edge, dim] (r3) -> (r4);
		
		
		\draw [edge, rounded corners] (init) -- ($(init) + (-0.6,-0.5)$) -- (c1);
		\draw [edge] (c1) -> (c2);
		\draw [edge] (c2) -> (c3);
		\draw [edge,rounded corners] (c3) -- ($(c3) + (1.2,0.5)$) -- ($(init) + (0.6,-0.5)$) -- (init) ;
		\draw [edge, rounded corners, dim] (l4) -- ($(l4) + (0.5, 0)$) -- (c3) ;
		\end{tikzpicture}
		\end{adjustbox}
		\label{fig:relaventsubset}
	}%
	\hspace{0.5cm}
	\subfloat[]{
		\begin{adjustbox}{width=0.3\textwidth}%
		\begin{tikzpicture}[node distance = {0.5cm and 1.2cm}, 
		v/.style = {draw, circle, thick, black!90}, 
		f/.style = {draw, thick, circle, accepting, black!90}, 
		dim/.style = {black!20},
		highlight/.style = {graph1, thick},
		edge/.style = {-Latex, black!90}, 
		initial distance=0.5cm, every initial by arrow/.style={-Latex,thick, black!90}, initial text={}
		]
		\node (init) [v,initial above] {};
		
		\node (l1) [v, below left = of init] {};
		\node (l2) [v, below = of l1] {};
		\node (l3) [v, dim, below = 1.4 of l2] {};
		\node (l4) [v, dim, below = of l3] {};
		
		\node (r1) [v, thick, graph1, below right = of init] {};
		\node (r2) [v, dim, below = of r1] {};
		\node (r3) [v, dim, below = 1.4 of r2] {};
		\node (r4) [v, dim, below = of r3] {};
		
		\node (c1) [v] at ($(init) + (-0.6,-2)$) {};
		\node (c2) [v, below = of c1] {};
		\node (c3) [f, below = of c2] {};
		
		\draw [edge] (init) -> (l1);
		\draw [edge] (l1) -> (l2);
		\draw [edge, dim] (l3) -> (l4);
		
		\draw [edge, highlight] (init) -> (r1);
		\draw [edge, dim] (r1) -> (r2);
		\draw [dotted, thick, dim] (r2) -> (r3);
		\draw [edge, dim] (r3) -> (r4);
		\node (x2) [v, thick, graph1] at ($(l2) + (0, -0.8)$) {};
		\draw [edge, highlight] (l2) -> (x2);
		
		\draw [dotted, dim] (x2) -> (l3);
		
		\draw [edge, rounded corners] (init) -- ($(init) + (-0.6,-0.5)$) -- (c1);
		\draw [edge] (c1) -> (c2);
		\draw [edge] (c2) -> (c3);
		\draw [edge,rounded corners] (c3) -- ($(c3) + (1.2,0.5)$) -- ($(init) + (0.6,-0.5)$) -- (init) ;
		\path  (x2) edge [in = 215, out=145, looseness =7, highlight] node {} (x2);
		\path  (r1) edge [in = -35, out=35, looseness =7, highlight] node {} (r1);
		\draw [edge,rounded corners, dim] (l4) -- ($(l4) + (0.5,0)$) -- (c3) ;
		\end{tikzpicture}
		\end{adjustbox}
		\label{fig:lowerbound}
	}%
	\hspace{0.5cm}
	\subfloat[]{
		\begin{adjustbox}{width=0.3\textwidth}%
		\begin{tikzpicture}[node distance = {0.5cm and 1.2cm}, 
		v/.style = {draw, circle, thick, black!90}, 
		f/.style = {draw, thick, circle, accepting, black!90}, 
		dim/.style = {black!20},
		edge/.style = {-Latex, black!90}, 
		highlight/.style = {graph1, thick},
		initial distance=0.5cm, every initial by arrow/.style={-Latex,thick, black!90}, initial text={}
		]
		\node (init) [v,initial above] {};
		
		\node (l1) [v, below left = of init] {};
		\node (l2) [v, below = of l1] {};
		\node (l3) [v, dim, below = 1.4 of l2] {};
		\node (l4) [v, dim, below = of l3] {};
		
		\node (r1) [f, very thick, highlight, below right = of init] {};
		\node (r2) [v, dim, below = of r1] {};
		\node (r3) [v, dim, below = 1.4 of r2] {};
		\node (r4) [v, dim, below = of r3] {};
		
		\node (c1) [v] at ($(init) + (-0.6,-2)$) {};
		\node (c2) [v, below = of c1] {};
		\node (c3) [f, below = of c2] {};
		
		\draw [edge] (init) -> (l1);
		\draw [edge] (l1) -> (l2);
		\draw [edge, dim] (l3) -> (l4);
		
		\draw [edge, highlight] (init) -> (r1);
		\draw [edge, dim] (r1) -> (r2);
		\draw [dotted, thick, dim] (r2) -> (r3);
		\draw [edge, dim] (r3) -> (r4);
		\node (x2) [f, v, thick, graph1] at ($(l2) + (0, -0.8)$) {};
		\draw [edge, highlight] (l2) -> (x2);
		
		\draw [dotted, dim] (x2) -> (l3);
		
		\draw [edge, rounded corners] (init) -- ($(init) + (-0.6,-0.5)$) -- (c1);
		\draw [edge] (c1) -> (c2);
		\draw [edge] (c2) -> (c3);
		\draw [edge,rounded corners] (c3) -- ($(c3) + (1.2,0.5)$) -- ($(init) + (0.6,-0.5)$) -- (init) ;
		\path  (x2) edge [in = 215, out=145, looseness =7, highlight] node {} (x2);
		\path  (r1) edge [in = -35, out=35, looseness =7, highlight] node {} (r1);
		\draw [edge,rounded corners,dim] (l4) -- ($(l4) + (0.5,0)$) -- (c3) ;
		\end{tikzpicture}
		\end{adjustbox}
		\label{fig:upperbound}
	}%
	\caption{Fig. \ref{fig:relaventsubset} depicts the relevant subset obtained at line \ref{alg:main:rel} of Algorithm \ref{alg:main}. Fig.~\ref{fig:lowerbound} and Fig.~\ref{fig:upperbound} show the addition of successors (in highlight) of the states at the fringe. In Fig.~\ref{fig:lowerbound}, the appended states are made absorbing by adding a self-loop of rate $\exitmax$. Meanwhile in Fig.~\ref{fig:upperbound}, the newly added states are made goals.}
\end{figure}
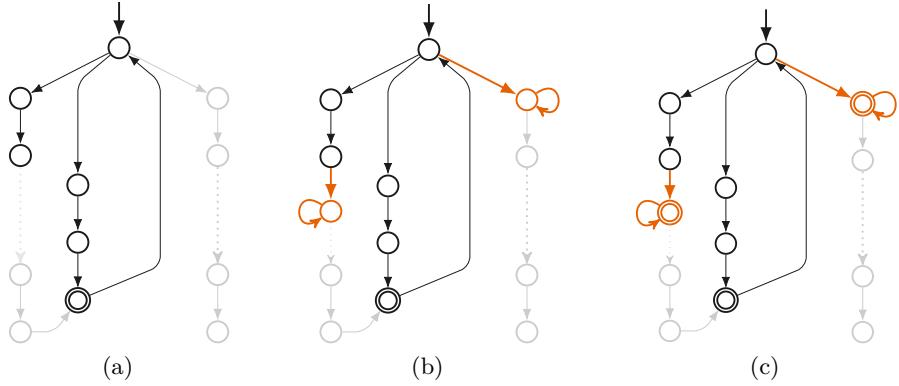

\begin{example}
Let $S'$ be the state space of the CTMDP from Figure \ref{fig:running-a} explored in Example \ref{epl:running}. Figure \ref{fig:relaventsubset} depicts the sub-CTMDP obtained by restricting the state space of the original model to $S'$. Figures \ref{fig:lowerbound} and \ref{fig:upperbound} demonstrate how the ``pessimistic'' and ``optimistic'' CTMDPs can be obtained. All the states that are not part of $S'$ are made absorbing for the ``pessimistic'' CTMDP (\ref{fig:lowerbound}) and are made goal states for the ``optimistic'' CTMDP (\ref{fig:upperbound}).
\end{example}

Formally, we define methods $\lowerSub{\Ctmdp}{S'}$ and $\upperSub{\Ctmdp}{S'}$ that return the pessimistic and optimistic CTMDP, respectively. The $\lowerSub{\Ctmdp}{S'}$ method returns a CTMDP $\underline{\Ctmdp} = (\initstate, \widetilde{S}, \Act, \widetilde{\bfR}, \goals)$, where $\widetilde{S} = S' \cup \Succ(S')$, and $\forall s', s'' \in \widetilde{S}$:
$$
\widetilde{\bfR}[s',\alpha, s''] = 
\left\lbrace
\begin{array}{llc}
\bfR[s',\alpha, s''] & \mbox{ if } s' \in S'\\
\exitmax & \mbox{ if } s' \not \in S', s'' = s' \\
0 & \mbox{ otherwise,}
\end{array}
\right.
$$ 
where $\exitmax$ is the maximum exit rate in $\Ctmdp$.
And the method $\upperSub{\Ctmdp}{S'}$ returns CTMDP $\overline{\Ctmdp} = (\initstate, \widetilde{S}, \Act, \widetilde{\bfR}, \overline{\goals})$, where $\overline{\goals} = \goals \cup (\widetilde{S} \setminus S')$, and state space $\widetilde{S}$ and the rate matrix $\widetilde{\bfR}$ are the same as for $\lowerSub{\Ctmdp}{S'}$.

Since many states are absorbing now large parts of the state space may become unreachable, namely all the states that are not in $\widetilde{S}$.

\begin{restatable}{lemma}{lowerUpperSub}
	\label{lem:lower-upper-sub}
	Let $\lowerSub{\Ctmdp}{S'} = \underline{\Ctmdp}$ and $\upperSub{\Ctmdp}{S'} = \overline{\Ctmdp}$, then
	$$\val{\underline{\Ctmdp}}{}{T} \leqslant \val{\Ctmdp}{}{T} \leqslant \val{\overline{\Ctmdp}}{}{T} $$
\end{restatable}

\ssecspace

\subsection{Step 3: Computing the Reachability Value}

\ssecspace

Algorithm \ref{alg:main} requires computing the reachability values for CTMDPs $\underline{\Ctmdp}$ and $\overline{\Ctmdp}$ (line \ref{alg:main:lower-upper-sch}). 
This can be done by any algorithm for reachability analysis, e.\,g. \cite{DBLP:conf/atva/ButkovaHHK15,DBLP:conf/qest/NeuhausserZ10,DBLP:conf/fsen/HatefiH13,DBLP:journals/cor/BuchholzS11,fearnley_et_al:LIPIcs:2011:3354,DBLP:conf/tacas/BaierHHK04} which approximate the value up to an arbitrary precision $\varepsilon$. 
These algorithms usually also compute the $\varepsilon$-optimal scheduler along with the approximation of the reachability value.
In the following we will use interchangeably the notions of the value and its $\varepsilon$-approximation, as well as an optimal scheduler and an $\varepsilon$-optimal scheduler.

Notice that some of the algorithms mentioned above compute optimal reachability value only w.\,r.\,t. a subclass of schedulers, rather than the full class $\straas$. In this case the result of Algorithm \ref{alg:main} will be the optimal reachability value with respect to this subclass and not class $\straas$.
\ssecspace

\subsection{Step 4: The Choice of Scheduler $\straasim$}\label{sec:sol:sched}
\ssecspace

At line \ref{alg:main:straasim} of Algorithm \ref{alg:main} the scheduler $\straasim$ is selected that is used in the subsequent iteration for refining the relevant subset of states.
We propose two ways of instantiating the function $\ChooseScheduler(\straauni,\straaopt)$, one with the uniform scheduler $\straauni$, and another with the scheduler $\straaopt$. Depending on the model, its goal states and the time bound one of the options may deliver smaller relevant subset than another:

\begin{example}

\begin{figure}[t]
	\centering
	\scalebox{0.75}{
		\centering
	\subfloat[]{
		\begin{tikzpicture}[node distance = {0.5cm and 1.2cm}, 
		v/.style = {draw, circle, thick, minimum size=0.7cm}, 
		relv/.style = {circle, thick}, 
		none/.style = {draw=none,fill=none,inner sep=0,outer sep=0}, 
		f/.style = {draw, thick, circle, accepting}, 
		edge/.style = {-Latex, thick}, 
		rel/.style = {}, 
		ris/.style = {}, 
		ev/.style = {font=\scriptsize}, 
		initial distance=0.5cm, every initial by arrow/.style={-Latex,thick}, initial text={}
		]
		\node (init) [v, initial above] {0};
		
		\node (a1) [v] at ($(init) + (-2,0)$) {$a_1$};
		\node (a2) [v, below = of a1] {$a_2$};
		\node (b1) [v, below = of a2] {$b_1$};
		\node (b2) [none, inner sep=7, below = of b1] {$\cdots$};
		\node (b11) [v, below = of b2] {$b_{11}$};

		\node (c1) [v] at ($(init) + (2,0)$) {$c_1$};
		\node (c2) [v, below = of c1] {$c_2$};
		\node (c3) [v, below = of c2] {$c_3$};
		\node (rc3) [none] at ($(c3) + (0,-2.4)$) {};

		\node (g) [f]  at ($(init) + (0,-4.85)$)  {g};
		
		\node (xa) [none] at ($(init)!0.5!(a1)$) {};
		\node (xb) [none] at ($(init)!0.5!(c1)$) {};

		\path 
			(init)	edge [midway, above, -, dashed, rel] node[] {$\alpha$} (xa)
					edge [midway, above, -, dashed, ris] node {$\beta$} (xb)	
			(xb)	edge [midway, above] node[ev] {$0.5$} (c1)
			(a1)	edge [midway, left] node[ev] {$0.51$} (a2)
			(a2)	edge [midway, left] node[ev] {$0.51$} (b1)				(b1)	edge [midway, left] node[ev] {$175$} (b2)
			(b2)	edge [midway, left] node[ev] {$175$} (b11)
			(b11)	edge [midway, above] node[ev] {$175$} (g)
			(xa)	edge [midway, above] node[ev] {$0.51$} (a1)
			(c1)	edge [midway, left] node[ev] {$0.5$} (c2)
			(c2)	edge [midway, left] node[ev] {$0.5$} (c3)
			(g)	edge [midway, left, loop above] node[] {$1$} (g)
			
	;
	\draw [edge, rounded corners, ris] (c3) node[ev,left]at($(c3)!0.5!(rc3)$){$0.5$} -- ($(rc3)$) -- (g);
		\end{tikzpicture}
		\label{fig:uni-better:a}
	}

\hspace*{2cm}

	\subfloat[]{
		\begin{tikzpicture}[node distance = {0.5cm and 1.2cm}, 
		v/.style = {draw, circle, thick, minimum size=0.7cm}, 
		relv/.style = {circle, thick}, 
		none/.style = {draw=none,fill=none,inner sep=0,outer sep=0}, 
		f/.style = {draw, thick, circle, accepting}, 
		edge/.style = {-Latex, thick}, 
		rel/.style = {}, 
		ris/.style = {}, 
		ev/.style = {font=\scriptsize}, 
		initial distance=0.5cm, every initial by arrow/.style={-Latex,thick}, initial text={}
		]
		\node (init) [v, initial above] {0};
		
		\node (a1) [v] at ($(init) + (-2,0)$) {$d_1$};
		\node (a2) [v, below = of a1] {$d_2$};
		\node (b1) [v, below = of a2] {$e_1$};
		\node (b2) [none, inner sep=7, below = of b1] {$\cdots$};
		\node (b11) [v, below = of b2] {$e_{11}$};

		\node (c1) [v] at ($(init) + (2,0)$) {$f_1$};
		\node (c2) [v, below = of c1] {$f_2$};
		\node (c3) [v, below = of c2] {$f_3$};
		\node (c4) [none, below = of c3] {$\cdots$};
		\node (c13)[v] at ($(c3) + (0,-2.4)$) {$f_{13}$};

		\node (g) [f]  at ($(init) + (0,-4.85)$)  {g};
		
		\node (xa) [none] at ($(init)!0.5!(a1)$) {};
		\node (xb) [none] at ($(init)!0.5!(c1)$) {};

		\path 
			(init)	edge [midway, above, -, dashed, rel] node[] {$\alpha$} (xa)
					edge [midway, above, -, dashed, ris] node {$\beta$} (xb)	
			(xb)	edge [midway, above] node[ev] {$0.5$} (c1)
			(a1)	edge [midway, left] node[ev] {$0.51$} (a2)
			(a2)	edge [midway, left] node[ev] {$0.51$} (b1)				(b1)	edge [midway, left] node[ev] {$175$} (b2)
			(b2)	edge [midway, left] node[ev] {$175$} (b11)
			(b11)	edge [midway, above] node[ev] {$175$} (g)
			(xa)	edge [midway, above] node[ev] {$0.51$} (a1)
			(c1)	edge [midway, left] node[ev] {$0.5$} (c2)
			(c2)	edge [midway, left] node[ev] {$0.5$} (c3)
			(c3)	edge [midway, left] node[ev] {$0.5$} (c4)
			(c4)	edge [midway, left] node[ev] {$0.5$} (c13)
			(c13)	edge [midway, above] node[ev] {$0.5$} (g)
			(g)	edge [midway, left, loop above] node[] {$1$} (g)
			
	;
		\end{tikzpicture}
		\label{fig:uni-better:b}
	}
	}
		\caption{}
\end{figure}

Consider, the CTMDP in Figure \ref{fig:uni-better:a} and the time bound $3.0$.
Assuming that the goal state has not yet been sampled from the right and left chains, action $\alpha$ delivers higher reachability value than action $\beta$.
For example, if states $a_1$ to $a_2$ are sampled from the chain on the left and $c_1$ to $c_2$ from the chain on the right,
the reachability value of the respective over-approximating CTMDP when choosing action $\beta$ is $0.1987$ and when choosing action $\alpha$ is $0.1911$.
And this situation persists also when states $b_1-b_{10}$ are sampled due to high exit rates of the respective transitions.
However if state $b_{11}$ is sampled, the reachability value when following $\alpha$ becomes $0.1906$.
Only at this moment the optimal behaviour is to choose action $\beta$.
However, when following the uniform scheduler, there is a chance that the whole chain on the right is explored before any of the states $b_i$ are visited.
If the precision $\varepsilon=0.01$, then at the moment the goal state is reached via the right chain and at least states $a_1$ to $a_2$ are sampled on the left, the algorithm has converged.
Thus using the uniform scheduler $\Algo$ may in fact explore fewer states than when using the optimal one.

Naturally, there are situation when following the optimal scheduler is the best one can do.
For example, in the CTMDP in Figure \ref{fig:uni-better:b} it is enough to explore only state $f_1$ on the right to realise that action $\beta$ is sub-optimal.
From this moment on only action $\alpha$ is chosen for simulations, which is in fact the best way to proceed.
At the moment the goal state is reached the algorithm has converged for precision $0.01$.

\end{example}

One of the main advantages of the uniform scheduler is that it does not require too much memory and is simple to implement.
Moreover, since some algorithms to compute time-bounded reachability probability do not provide an optimal scheduler in the classical way as defined in Section \ref{sec:ctmcp} (\cite{DBLP:conf/atva/ButkovaHHK15}), the use of $\straauni$ may be the only option.
In spite of its simplicity, in many cases this scheduler generates very succinct state spaces, as we will show in Section \ref{sec:experiments}. 

Using the uniform scheduler is beneficial in those cases when, for example, different actions of the same state have exit rates that differ drastically, e.\,g. by an order of magnitude. If the goal state is reachable via actions with high rates, choosing an action with low rate leads to higher residence times (due to properties of the exponential distribution) and therefore fewer states will be reachable within the time bound, compared to choosing an action with a high exit rate. In this case using the uniform scheduler may lead to larger sub-space, compared to using the optimal scheduler. However, the experiments show this difference is typically negligible.

The drawback of the uniform scheduler is that the probability of it choosing each action is positive. 
Thus it will choose also those actions that are clearly suboptimal and could be omitted during the simulations.
The uniform scheduler $\straauni$ does not take this information into account while the scheduler $\straaopt$ does. The latter is optimal on the sub-CTMDP obtained during the previous iterations. This scheduler will thus pick only those actions that look most promising to be optimal. Using this scheduler may induce smaller sampled state space than the one generated by $\straauni$, as we also show in Section \ref{sec:experiments}.

Notice that it is possible to alternate between using $\straauni$ and $\straaopt$ at different iterations of Algorithm \ref{alg:main}, for instance, when $\straaopt$ is costly to obtain or simulate.
However, in our experiments, we always choose either one of the two, with the exception for the first iteration when only the uniform scheduler is available.

\subsection{Step 5: Termination and Optimal Schedulers} 
\ssecspace

The algorithm runs as long as the reachability values of $\underline{\Ctmdp}$ and $\overline{\Ctmdp}$, as computed in Step 3 are not sufficiently close. It terminates when the difference becomes less than $\varepsilon$. The scheduler $\lstraaopt$ obtained in line \ref{alg:main:lower-upper-sch} of Algorithm \ref{alg:main} is $\varepsilon$-optimal for $\underline{\Ctmdp}$ since it is obtained by running a standard TBR algorithm on $\underline{\Ctmdp}$.
From this scheduler one can obtain $\varepsilon$-optimal scheduler $\straa$ for $\Ctmdp$ itself by choosing the same actions as $\lstraaopt$ on the relevant subset of states ($S'$ in Algorithm \ref{alg:main}) and any arbitrary action on states that are not relevant.
By using this extended scheduler on $\Ctmdp$, a value of $\val{\underline{\Ctmdp}}{}{T}$ can be achieved at the least. On the other hand, the scheduler $\straaopt$ is $\varepsilon$-optimal for $\overline{\Ctmdp}$. This naturally provides an upper bound on the reachability value which can be obtained in $\Ctmdp$.
Therefore the value of $\straa$ lies within $\left[\val{\underline{\Ctmdp}}{}{T}, \val{\overline{\Ctmdp}}{}{T}\right]$, which makes $\straa$ $\varepsilon$-optimal for $\Ctmdp$.

\begin{restatable}{lemma}{optScheduler}
  \label{lem:opt-sch}
  Scheduler $\straa$ computed by Algorithm \ref{alg:main} is $\varepsilon$-optimal.
\end{restatable}

\begin{restatable}{theorem}{correctness}
  \label{thm:correctness}
  Algorithm \ref{alg:main} converges almost surely.
\end{restatable}

On any CTMDP, if $\straasim = \straauni$, Algorithm \ref{alg:main} will, in the worst case, eventually explore the whole CTMDP. In such a situation, $\underline{\Ctmdp}$ and $\overline{\Ctmdp}$ will be the same as $\Ctmdp$. The algorithm would then terminate since the condition on line \ref{alg:main:condn} would be falsified. If $\straasim = \straaopt$, the system is continuously driven to the fringe as long as the condition on line \ref{alg:main:condn} holds. This is because all unexplored states act as goal states in the upper-bound model. Such a scheduler will eventually explore the state-space reachable by the optimal scheduler on the original model and leave out those parts that are only reachable with suboptimal decisions.

\secspace

\section{Experiments}\label{sec:experiments}

\ssecspace

The framework described in Section \ref{sec:algo} was evaluated against 5 different benchmarks available in the MAPA language \cite{mapa}: 

\begin{description}
	\item[\textbf{Fault Tolerant Work Station Cluster}  (\ftwc-n)] \cite{DBLP:conf/srds/HaverkortHK00}: models two networks of $n$ workstations each. Each network is interconnected by a switch. The two switches communicate via a backbone. All the components may fail and can be repaired only one at a time. 
	The system starts in a fully functioning state and a state is goal if in both networks either all the workstations or the switch are broken.
	\item[\textbf{Google File System} (\gfs-n)] \cite{DBLP:conf/dsn/HaverkortCHK02,DBLP:conf/sosp/GhemawatGL03}: in this benchmark files are split into chunks of equal size, each chunk is maintained by one of $n$ chunk servers. We fix the number of chunks a server may store to 5000 and the total number of chunks to 10000. 
	The GFS starts in the broken state where no chunk is stored. A state is defined as goal if the system is back up and for each chunk at least 3 copies are available.
	\item[\textbf{Polling System} (\ps-j-k-g):] We consider the variation of the polling system case \cite{DBLP:conf/qest/GuckHHKT13} \cite{DBLP:conf/formats/TimmerPS13}, that consists of $j$ stations and one server. Incoming requests of $j$ types are buffered in queues of size $k$ each, until they are processed by the server and delivered to their station. The system starts in a state with all the queues being nearly full. We consider 2 goal conditions: (i) all the queues are empty (g=\all) and (ii) one of the queues is empty (g=\one).
	\item[\textbf{Erlang Stages} (\erlang-k-r):] this is a synthetic model with known characteristics \cite{DBLP:conf/tacas/ZhangN10}.
It has two different paths to reach the goal state: a fast but risky path or a slow but sure path. The slow path is an Erlang chain of length k and rate r.
	\item[\textbf{Stochastic Job Scheduling Problem} (\sjs-m-j)] \cite{DBLP:journals/jacm/BrunoDF81}: models a multiprocessor architecture running a sequence of independent jobs. It consists of $m$ identical processors and $j$ jobs. As goal we define the states with all jobs completed;
\end{description}

By setting different model parameters for each of these benchmarks, we were able to generate models ranging from hundreds to millions of states. We used the tool \scoop\, \cite{scoop} to instantiate and convert the MAPA models into explicit state space CTMDPs.  

Our algorithm is implemented as an extension to \prism\, \cite{KNP11} and we use \imca\, \cite{imca} in order to solve the sub-CTMDPs ($\underline{\Ctmdp}$ and $\overline{\Ctmdp}$).
We would like to remark, however, that the performance of our algorithm can be improved by using a better toolchain than our \prism-\imca\ setup (see Appendix \ref{sec:app:experiements}).


In order to instantiate our framework, we need to describe how we perform Steps 1 and 3 (Section \ref{sec:algo:solution}). Recall from Section \ref{sec:rel-subset} that we proposed two different schedulers to be used as the simulating scheduler $\straasim$
: the uniform scheduler $\straauni$ and the optimal scheduler $\straaopt$ obtained by solving $\overline\Ctmdp$.

For Step 3, we select three algorithms for time-bounded reachability analysis: the first discretisation-based algorithm \cite{DBLP:conf/qest/NeuhausserZ10} (\tbn),  and the two most competitive algorithms according to the comparison performed in \cite{DBLP:conf/atva/ButkovaHHK15}, namely the adaptive version of discretization \cite{DBLP:journals/cor/BuchholzS11} (\tbb) and the uniformisation-based \cite{DBLP:conf/atva/ButkovaHHK15} (\tbu). $\Algo$ instantiated with these algorithms and with $\straasim = \straauni$ is referred to with \itbn, \itbb\, and \itbu\, respectively. 
For $\straasim = \straaopt$, the instantiations are referred to as \itbnplus, \itbbplus\ and  \itbuplus. Since  \tbu ~does not provide the scheduler in a classical form as defined in Section \ref{sec:prelim}, we omit \itbuplus. 
We also omit experiments on \itbnplus\ as our experience with \tbn\ and \itbn\ suggested that \itbnplus\ would also run out of time on most experiments.
 
We compare the performance of the instantiated algorithms with their originals, implemented in \imca. We set the precision parameter for $\Algo$ and the original algorithms in \imca\, to $0.01$. Indicators such as the median model checking time (excluding the time taken to load the model into memory) and explored state-space are measured. More details about the experimental setup are available in the Appendix \ref{sec:app:experiements}.

\begin{table}[t]
	\centering
	\caption{An overview of the experimental results along with the state-space sizes. Runtime (in seconds) for the various algorithms are presented. For more details on the experimental setup, see Appendix \ref{sec:app:experiements}. `-' indicates a timeout (1800 secs). \itbu, \itbb\ and \itbbplus\ perform quite well on \erlang, \gfs\ and \ftwc\ while only \itbbplus\ is better than \tbu\ and \tbb\ on the \ps-\one\ family of models. \ps-4-8-\all\ and \sjs\ are hard instances for both $\straauni$ and $\straaopt$. \tbn\ times out on all benchmarks except on \sjs\ because of its small state-space.}
	\resizebox{0.96\columnwidth}{!}{%
	\makebox[\textwidth][c]{
	{\renewcommand{\arraystretch}{1.3}%
		\setlength{\tabcolsep}{8pt}
	\begin{tabular}{lr@{\hskip 24pt}rr@{\hskip 24pt}rrr@{\hskip 24pt}rr}
		\toprule
		Benchmark & States & \tbu & \itbu & \tbb & \itbb & \itbbplus & \tbn & \itbn \\
		\midrule
\erlang-$10^{6}$-10 & 1,000k & 71 & \textbf{1} & 4 & \textbf{1} & \textbf{1} & - & 299 \\
\gfs-120 & 1,479k & - & \textbf{2} & - & \textbf{2} & \textbf{2} & - & - \\
\ftwc-128 & 597k & 251 & \textbf{10} & 114 & \textbf{11} & \textbf{15} & - & - \\
\ps-4-24-\one & 7,562k & 507 & - & 171 & - & \textbf{105} & - & - \\
\ps-4-8-\all & 119k & 1,475 & - & 826 & - & - & - & - \\
\sjs-2-9 & 18k & 6 & 99 & 2 & 139 & - & 1,199 & - \\ 
		\bottomrule
		\vspace{-6pt}
	\end{tabular}
	}
	}
	}
	\label{tab:overall}
\end{table}

\begin{table}[t]
	\centering
	\caption{For each benchmark, we report 
		(i) the size of the state-space; 
		(ii) total states explored by our instantiations of $\Algo$ until convergence; 
		(iii) size of the final over-approximating sub-CTMDP $\overline{\Ctmdp}$; and 
		(iv) the number of states which need to be kept as returned by running the greedy search of Section \ref{sec:expr:greedy} for smallest sub-CTMDP. 
		We use \ps-4-4-\one\, and \sjs-2-7 instead of larger models in their respective families as running the greedy search is a highly computation-intensive task.}
	\makebox[\textwidth][c]{
	{\renewcommand{\arraystretch}{1.3}%
				\setlength{\tabcolsep}{12pt}
	\begin{tabular}{lrrrrrrrr}
		\toprule
		& & \multicolumn{2}{c}{Explored} & & \\
		\cmidrule{3-4}
		\thead{Benchmark} & \thead{States} & \thead{by $\straasim$} & \thead{\%} & \thead{Size of\\last $\overline\Ctmdp$} & \thead{Post greedy\\reduction} \\
		\midrule
\erlang-$10^{6}$-10 & 1,000k & 559 & 0.06 & 561 & 496 \\
\gfs-120 & 1,479k & 105 & 0.01 & 200 & 85 \\
\ftwc-128 & 597k & 296 & 0.05 & 858 & 253 \\
\sjs-2-7 & 2k & 2,537 & 93.86 & 2,704 & 1,543 \\
\ps-4-4-\one & 10k & 697 & 6.63 & 2,040 & 696 \\
\ps-4-8-\all & 119k & - & - & - & - \\
\ps-4-24-\one & 7,562k & 23,309 & 0.31 & - & - \\ 
		\bottomrule
		\vspace{-6pt}
	\end{tabular}
	}
	}
\label{tab:explored}
\end{table}

Tables \ref{tab:overall} and \ref{tab:explored} summarize the main results of our experiments. 
Table \ref{tab:overall} reports the running time of the algorithms on several benchmarks, 
while Table \ref{tab:explored} reports on 
the size of the state-space of the models, 
the states explored by $\Algo$, 
the size of the over-approximating sub-CTMDP $\overline{\Ctmdp}$ when the algorithm terminates and 
the smallest relevant subset of $\overline{\Ctmdp}$ that we can obtain with 
reasonable effort. This subset is computed by a greedy algorithm described in Section \ref{sec:expr:greedy}.
It attempts to reduce more states of the explored subset without sacrificing 
the precision too much. We run the greedy algorithm with a precision of 
$\varepsilon/10$, where $\varepsilon$ is the precision used in $\Algo$.

We recall that our framework is targeted towards models which contain a small subset of valuable states. We can categorize the models into three classes:

\begin{description}
	\item[Easy with Uniform Scheduler ($\straasim=\straauni$).] Surprisingly enough, the uniform scheduler performs well on many instances, for example \erlang, \gfs\, and \ftwc. 
	For \erlang ~and \gfs, it was sufficient to explore a few hundred states no matter how the parameter which increased the state-space was changed (see description of the models above).
	Here the running time of the instantiations of our framework outperformed the original algorithms due to the fact that less than 1\% of the state-space is sufficient to approximate the reachability value up to precision $0.01$.
	\item[Easy with Optimal Scheduler ($\straasim=\straaopt$).] 

Predictably, there are cases in which uniform scheduler does not provide good results. For example consider the case of \ps-4-24-\one. Here the goal condition requires that one of the queues be empty. 
An action in this benchmark determines the queue from which the task to be processed is picked. Choosing tasks uniformly from different queues, not surprisingly, leads to larger explored state spaces and longer runtimes. 
Notice that all the instantiations that use uniform scheduler run out of time on this instance. On the other hand, targeted exploration with the most promising scheduler 
(column \itbbplus) performs even better than the original algorithm \tbb, finishing within 105~s compared to 171~s and exploring only 0.31\% of the state space. 
	\item[Hard Instances.] Naturally there are instances where it is not possible to find a small sub-CTMDP that preserves the properties of interest.
	For example in \ps-4-8-\all, the system is started with all queues being nearly full and the property queried requires all of the queues in the polling system to be empty. 
	As discussed in the beginning of Section \ref{sec:algo}, most of the states of the model have to be explored in order to reach the goal state. 
	In this model there is simply no small sub-CTMDP that preserves the reachability probabilities. 
	As expected, all instantiations timed out and nearly all the states had to be explored. 
	The situation is similar with \sjs.
	We identified (using the greedy algorithm in Section \ref{sec:expr:greedy}) that on some small instances of this model, only 30\% to 40\% of the state-space can be sacrificed.
	\item[Explored State Space and Running Time.] In general, as we have mentioned in Section \ref{sec:algo}, the problem is heavily dependent not only on the structure of the model,
	but also on the specified time-bound and the goal set. Increasing the time-bound for \erlang, for example, leads to higher probability to explore fully the states of the Erlang chain. 
	This is turns affects the optimal scheduler and for some time-bounds no small sub-CTMDP preserving the reachability value exists.

	Naturally, whenever the algorithm explored only a small fraction of the state space, the running time was usually also smaller than the running time of the respective original algorithm. 
	The performance of our framework is heavily dependent on the parameter $\nSim$. This is due to the fact that computation of the reachability value is an expensive operation when performed many times even on small models. 
	Usually in our experiments the amount of simulations was in the order of several thousands. For more details please refer to Appendix \ref{sec:app:experiements}.

\end{description}

\ssecspace

\subsection{Greedy Search for the Smallest sub-CTMDP} \label{sec:expr:greedy}

In this section, we provide an argument that in the cases where our techniques do not perform well, the reason is not a poor choice of the relevant subsets, but rather that in such cases there are no small subsets which can be removed, at least not such that can be easily obtained.
An ideal brute-force method to ascertain this would be to enumerate all subsets of the state space, make the states of the subset absorbing ($\underline{\Ctmdp}$) or goal ($\overline{\Ctmdp}$) and then to check whether the difference in values of $\underline{\Ctmdp}$ and $\overline{\Ctmdp}$ is $\varepsilon$-close only for small subsets.
Unfortunately, this is computationally infeasible.
As an alternative, we now suggest a greedy algorithm which we use to search for the largest subset of states one could remove in reasonable time.
The results of running this algorithm is presented in the right-most column of Table \ref{tab:explored}.

The idea is to systematically pick states and observe their effect on the value when they are made absorbing ($\underline{\Ctmdp}(s)$) or goal ($\overline{\Ctmdp}(s)$).
If a state does not influence the value of the original CTMDP too much, then $\delta(s) = \val{\overline{\Ctmdp}(s)}{}{T}-\val{\underline{\Ctmdp}(s)}{}{T}$ would be small.
We first sort all the states in ascending order according to the value $\delta(s)$.
And then iteratively build $\underline{\Ctmdp}$ and $\overline{\Ctmdp}$ by greedily picking states in this order and making them absorbing (for $\underline{\Ctmdp}$) and goal (for $\overline{\Ctmdp}$).
The process is repeated until $\val{\overline{\Ctmdp}}{}{T}-\val{\underline{\Ctmdp}}{}{T}$ exceeds $\varepsilon$.

\ssecspace

\section{Conclusion}

\ssecspace

We have introduced a framework for time-bounded reachability analysis of continuous-time Markov decision processes.
This framework allows us to run arbitrary algorithms from the literature on a subspace of the original system and thus obtain the result faster, while not compromising its precision beyond a given $\varepsilon$.
The subspace is iteratively identified using simulations.
In contrast to the standard algorithms, the amount of computation needed reflects not only the model, but also the property to be checked.

The experimental results have revealed that the models often have a small subset which is sufficient for the analysis, and thus our framework speeds up all three considered algorithms from the literature.
For the exploration, already the uninformed uniform scheduler proves efficient in many settings.
However, the more informed scheduler, fed back from the analysis tools, may provide yet better results.
In cases where our technique explores the whole state space, our conjecture, confirmed by the preliminary results using the greedy algorithm, is that these models actually do not posses any small enough relevant subset of states and cannot be exploited by this approach.

This work is agnostic of the structure of the models.
Given that states are typically given by a valuation of variables, the corresponding structure could be further utilized in the search for the small relevant subset.
A step in this direction could follow the ideas of \cite{DBLP:conf/icse/PaveseBU13}, where discrete-time Markov chains are simulated, the simulations used to infer invariants for the visited states, and then the invariants used to identify a subspace of the original system, which is finally analyzed.
An extension of this approach to a non-deterministic and continuous setting could speed up the subspace-identification part of our approach and thus decrease our overhead.
Another way to speed up this process is to quickly obtain good schedulers (with no guarantees), e.g. \cite{DBLP:journals/pe/BartocciBBMS17}, use them to identify the subspace faster and only then apply a guaranteed algorithm.

\bibliographystyle{alpha-abbrev}

\bibliography{ref}

\appendix
\section{Appendix}
\label{sec:appendix}
\subsection{Proofs}
\lowerUpperSub*
\begin{proof}
Let $s$ be a state of a CTMDP. Then by definition:

\begin{align*}
	\val{}{s}{T} = \mathop{\max}\limits_{\alpha \in \Act(s)} \left\{ \int\limits_{0}^{T} \exit{s,\alpha}e^{-\exit{s,\alpha}t} \sum\limits_{s' \in S} \trans(s, \alpha, s') \cdot \val{}{s'}{T-t} \der{t} \right\}
\end{align*}

Due to the properties of $\max$ operator and integrals, for any function $f:\Realsplus \to \Realsplus$, s.\,t. $\forall t\in\Realsplus:~\val{}{s'}{t} \leqslant f(t)$ the following holds:

\begin{align}\label{fm:less}
	\val{}{s}{T} \leqslant \mathop{\max}\limits_{\alpha \in \Act(s)} \left\{ \int\limits_{0}^{T} \exit{s,\alpha}e^{-\exit{s,\alpha}t} \sum\limits_{s' \in S} \trans(s, \alpha, s') \cdot f(T-t) \der{t} \right\}
\end{align}

	The transformations of $\lowerSub{\Ctmdp}{S'}$ and $\upperSub{\Ctmdp}{S'}$ only affect those states that have at least one successor not in $S'$. Consider one of such states $s \in S'$, s.\,t. $\exists \alpha\in \Act(s), s' \in \Succ(s,\alpha) \cap (\widetilde{S} \setminus S')$. The functions $\lowerSub{\Ctmdp}{S'}$ and $\upperSub{\Ctmdp}{S'}$ make the state $s$ absorbing. If $s \in \goals$, this transformation does not affect the reachability value. If $s \not \in \goals$, then the value function of $s$ after the transformation by $\lowerSub{\Ctmdp}{S'}$ is a constant 0, and by $\upperSub{\Ctmdp}{S'}$ -- constant 1 (because the state becomes a new goal state). Since $\forall t \in \Realsplus:0 \leqslant \val{}{s'}{t} \leqslant 1$, then due to (\ref{fm:less}) the statement of the lemma follows.	
\qed
\end{proof}

\optScheduler*
\begin{proof}

We denote with $\val{\Ctmdp}{s,\sigma}{T}$ the reachability value achieved in $\Ctmdp$ under scheduler $\sigma$ starting from state $s$. Let $\straa$ be the scheduler produced by Algorithm \ref{alg:main}.

We will prove that $\val{\underline{\Ctmdp}}{}{T} \leqslant \val{\Ctmdp}{\straa}{T} \leqslant \val{\overline{\Ctmdp}}{}{T}$. 
First of all, due to Lemma \ref{lem:lower-upper-sub}: $\val{\Ctmdp}{\straa}{T} \leqslant \val{\Ctmdp}{}{T} \leqslant \val{\overline{\Ctmdp}}{}{T}$. 

We will prove now the other inequality. 
For simplicity, we consider CTMDP $\underline{\Ctmdp}$ to have the same state space and set of goal states as the state space $S$ and goal set $\goals$ of the original model. We do not modify any transition in $\underline{\Ctmdp}$. Due to the fact that the appended states are unreachable, this transformation does not affect the outcome of Algorithm \ref{alg:main}, it still produces the same values and sets of relevant states.

Let $\widetilde{S} = S' \cup \Succ(S')$, where $S'$ is the set of relevant states computed by Algorithm \ref{alg:main}.
We define $\val{\Ctmdp}{s,\sigma}{T, N}$ to be the reachability value from state $s$ for scheduler $\sigma$ and given that not more than $N$ transitions can be taken. Then $\val{\Ctmdp}{s,\sigma}{T} = \lim_{N\to\infty} \val{\Ctmdp}{s,\sigma}{T,N}$. We will prove by induction that for all $N \in \mathbb{N}_{\geqslant 0}, s \in S, T \in \mathbb{R}_{>0}$:

$$
\val{\underline{\Ctmdp}}{s,\lstraaopt}{T,N} \leqslant \val{\Ctmdp}{s,\straa}{T,N},
$$
where $\lstraaopt$ is the optimal scheduler for $\val{\underline{\Ctmdp}}{}{T}$. 

\begin{itemize}
	\item[$N=0$:] Since the state space of $\underline{\Ctmdp}$ and $\Ctmdp$ coincide, as well as the set of goal states, then obviously
	\begin{align*}
		\forall s \not \in \goals: &\val{\underline{\Ctmdp}}{s,\lstraaopt}{T,0} = \val{\Ctmdp}{s,\straa}{T,0} = 0\\
		\forall s \in \goals: &\val{\underline{\Ctmdp}}{s,\lstraaopt}{T,0} = \val{\Ctmdp}{s,\straa}{T,0} = 1\\
	\end{align*}
	
	\item[$N>0$:] For $s \not \in \widetilde{S}, s \in \goals: \val{\underline{\Ctmdp}}{s,\lstraaopt}{T,N} = \val{\Ctmdp}{s,\straa}{T,N} = 1$. For $s \not \in \widetilde{S}, s \not \in \goals:~\val{\underline{\Ctmdp}}{s,\lstraaopt}{T,N} = 0 \leqslant \val{\Ctmdp}{s,\straa}{T,N}$. Let $s \in \widetilde{S}$, we denote with $\trans_{\Ctmdp}(s, \alpha, s')$ the discrete transition relation in the CTMDP $\Ctmdp$. By definition of the reachability value:
	
\begin{alignat*}{2}
	\val{\underline{\Ctmdp}}{s,\lstraaopt}{T,N} 
	&= \int\limits_{0}^{T} \exit{s,\alpha}e^{-\exit{s,\alpha}t} &&\sum\limits_{s' \in S} \trans_{\underline{\Ctmdp}}(s, \alpha, s') \cdot \val{\underline{\Ctmdp}}{s', \lstraaopt}{T-t,N-1} \der{t}\\
	&= \int\limits_{0}^{T} \exit{s,\alpha}e^{-\exit{s,\alpha}t}
		\Big( 	&&\sum\limits_{s' \in \widetilde{S}} \underbracket{\trans_{\underline{\Ctmdp}}(s, \alpha, s')}_{=\trans_{\Ctmdp}(s, \alpha, s')} \cdot \underbracket{\val{\underline{\Ctmdp}}{s',\lstraaopt}{T-t, N-1}}_{\text{IH: } \leqslant \val{\Ctmdp}{s',\straa}{T-t, N-1}}  + \\
				& &&\underbracket{\sum\limits_{s' \not \in \widetilde{S}} \trans_{\underline{\Ctmdp}}(s, \alpha, s') \cdot \val{\underline{\Ctmdp}}{s',\lstraaopt}{T-t, N-1}}_{=0}
		\Big) \der{t} \\
	&\leqslant \int\limits_{0}^{T} \exit{s,\alpha}e^{-\exit{s,\alpha}t}
		\Big( 	&&\sum\limits_{s' \in \widetilde{S}} \trans_{\Ctmdp}(s, \alpha, s') \cdot \val{\Ctmdp}{s',\straa}{T-t, N-1}  + \\
				& &&\underbracket{\sum\limits_{s' \not \in \widetilde{S}} \trans_{\Ctmdp}(s, \alpha, s') \cdot \val{\Ctmdp}{s',\straa}{T-t, N-1}}_{\geqslant 0}
		\Big) \der{t} \\
	&= \val{\Ctmdp}{s,\straa}{T,N}
\end{alignat*}
	
\end{itemize}

\end{proof}

\correctness*
\begin{proof}

We at first argue about the correctness of the algorithm w.\,r.\,t. the instantiation $\ChooseScheduler(\straa) = \straa$.

First of all, let us notice that time-bounded reachability problem for CTMDP can be approximated up to arbitrarily small $\varepsilon$ by step-bounded reachability for discrete time MDP \cite{DBLP:conf/qest/NeuhausserZ10}. This is achieved by the so called discretisation approach. 

Given this discrete MDP, its step-bounded reachability can be computed by the algorithm from \cite{atva}. Notice that the back-propagation of values over a path in this algorithm is the same as running a classical MDP reachability algorithm \cite{Puterman} on 2 MDPs: for the lower bound all the states outside of the path are made absorbing, and for the upper bound - they are made goal states. Therefore one could as well sample several paths and back-propagate the values over MDPs where all states outside of the sampled states are made absorbing, or goal. This is the discrete analog of steps \ref{alg:main:lower-upper}-\ref{alg:main:straasim} of Algorithm \ref{alg:main}. With $\varepsilon \to 0$ the limiting behaviour of the algorithm coincides with Algorithm \ref{alg:main}. Given that \cite{atva} converges almost surely, Algorithm \ref{alg:main} as well converges almost surely.

We will now prove the correctness of the algorithm w.\,r.\,t. the instantiation $\ChooseScheduler(\straa) = \straauni$.

When the relevant subset is obtained with $\straauni$, an action is picked with uniform probability and the the next state is sampled according to the respective distribution. As step \ref{alg:main:rel} of Algorithm \ref{alg:main} may potentially be run infinitely often, the uniform sampling would eventually cover all states reachable from the initial state. In such a case, $S' = S$ and hence successors of $S'$ are already included in $S'$. Hence, the $\lowerSub{\Ctmdp}{S'}$ and  $\upperSub{\Ctmdp}{S'}$ trivially return $\Ctmdp$. Therefore, $\val{\underline{\Ctmdp}}{}{T} = \val{\Ctmdp}{}{T} = \val{\overline{\Ctmdp}}{}{T}$.
\qed
\end{proof}

\subsection{Details of Experiments} \label{sec:app:experiements}

\paragraph{Experimental Setup.} The experiments were on a multi-core Intel Xeon Server with sufficient RAM even though our experiments run only on a single core. Our algorithms are implemented in \prism\ and we use \imca\ in order to solve the partial models. \prism\ leverages the multiple cores to parse and load the model, after which the model checking happens only on a single core. \imca\ on the other hand runs completely on a single core. 
In order to make the comparison fairer, we measure only the time taken for performing the computations and not the time taken to parse and load the model file by the two tools. \prism\ is allotted 8GB of memory while no restriction is set for \imca. For each parameter configuration reported in the paper, the experiment is run at least 5 times and the median of the 5 runs are presented for runtime as well as the explored states.

\paragraph{Other Parameters: T and $\nSim$.} These parameters were briefly explained in Section \ref{sec:experiments}. Here we discuss the practical implications and the choice of the parameters. We choose time bound, T in such a way that that the probability of reaching the goals in each model is non-trivial (i.e. neither 0 nor 1). 
In Appendix \ref{sec:app:more-graphs}, we mention some of the time bounds we used. For the remaining models, \sjs-2-9 and \ps-4-8-\all, we used T=2 and T=1000 respectively. Another parameter which had to be set manually was the number of samples chosen, $\nSim$ in Algorithm \ref{alg:get-relevant}. 
This was chosen in such a way that \prism\ didn't query \imca\ for solving sub-CTMDPs too many times, since it was a major time-consumer in the chain. For  \sjs-2-9 and \ps-4-8-\all, $\nSim$ was chosen to be 40,000 and 20,000 respectively.

\paragraph{Toolchain Efficiency.} As hinted in the Section \ref{sec:app:experiements}, our prototype toolchain is far from being efficient. Whenever \prism\ wants a sub-CTMDP solved, it writes it into a file and calls \imca\ on it. Hence, every computation of the value (Algorithm \ref{alg:main}, line \ref{alg:main:lower-upper-sch}) is accompanied by the overhead of writing the file and loading the sub-CTMDP in \imca. We expect that a direct and continuous communication channel between the two tools would be able to save non-trivial amount of time.

\subsection{Additional experimental evaluation of scalability} \label{sec:app:more-graphs}

\begin{figure}
		\centering
	\begin{minipage}{0.45\textwidth}
	{\renewcommand{\arraystretch}{1.5}%
		\setlength{\tabcolsep}{8pt}
		\begin{tabular}{rrrrr}
			\toprule
			& & \multicolumn{2}{c}{Explored} \\
			\cmidrule{3-4}
			\thead{k} & \thead{States} & \thead{$\straauni$} & \thead{\%} \\ \midrule
			50000 & 500k & 553 & 1.11 \\
			100000 & 100k & 569 & 0.57 \\
			250000 & 250k & 562 & 0.22 \\
			500000 & 500k & 568 & 0.11 \\
			1000000 & 1,000k & 564 & 0.06 \\ 
			\bottomrule
		\end{tabular}
	}
\end{minipage}
	\hfill
	\pgfplotstableread{figures/erlang.data}
	\datatable
	\begin{minipage}{0.45\textwidth}
		\centering
		\subfloat{
			\begin{tikzpicture}
			\begin{axis}[
			height=5.5cm,
			ymode=log,
			log ticks with fixed point,
			xlabel={k ($\times10^3$)},
			ylabel={Time (s)},
			xmin=0, xmax=1000,
			ymin=0, ymax=100,
			xtick={100,250,500,1000},
			ytick={1, 10, 100},
			legend pos=outer north east,
			ymajorgrids=true,
			grid style=dotted,
			]
			
			\addplot [color = graph1, mark=*] table [y = TBU] from \datatable;
			\addlegendentry{\tbu}
			\addplot [color = graph2, mark=square*] table [y = TBB] from \datatable;
			\addlegendentry{\tbb}
			\addplot [color = graph3, mark=triangle*] table [y = M6-TBU] from \datatable;
			\addlegendentry{\itbu}
			\addplot [color = graph4, mark=diamond*] table [y = M6-TBB] from \datatable;
			\addlegendentry{\itbb}
			
			\end{axis}
			\end{tikzpicture}%
		}
	\end{minipage}
	\caption{Erlang Stages (T=50, $\nSim=1000$)}
\end{figure}
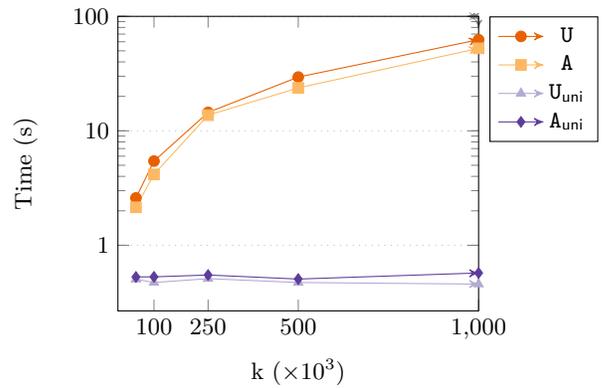

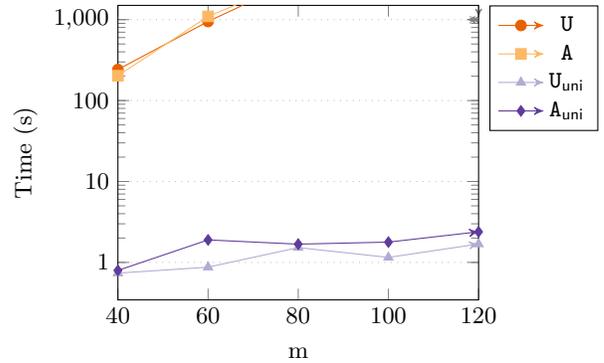
\begin{figure}
	\centering
\begin{minipage}{0.45\textwidth}
		{\renewcommand{\arraystretch}{1.5}%
							\setlength{\tabcolsep}{11pt}
			\begin{tabular}{rrrrr}
				\toprule
				& & \multicolumn{2}{c}{Explored} \\
				\cmidrule{3-4}
				\thead{m} & \thead{States} & \thead{$\straauni$} & \thead{\%} \\ \midrule
				40 & 166k & 99 & 0.06 \\
				60 & 372k & 102 & 0.03 \\
				80 & 659k & 119 & 0.02 \\
				100 & 1,029k & 107 & 0.01 \\
				120 & 1,479k & 100 & 0.01 \\ 
				\bottomrule
			\end{tabular}
		}
	\end{minipage}
	\hfill
	\pgfplotstableread{figures/gfs.data}
	\datatable
	\begin{minipage}{0.45\textwidth}
		\centering
		\subfloat{
			\begin{tikzpicture}
			\begin{axis}[
			ymode=log,
			log ticks with fixed point,
			height=5.5cm,
			xlabel={m},
			ylabel={Time (s)},
			xmin=40, xmax=120,
			ymin=0, ymax=1500,
			xtick={40, 60, 80, 100, 120},
			ytick={1, 10, 100, 1000},
			legend pos= outer north east,
			ymajorgrids=true,
			grid style=dotted,
			]
			
			\addplot [color = graph1, mark=*] table [y = TBU] from \datatable;
			\addlegendentry{\tbu}
			\addplot [color = graph2, mark=square*] table [y = TBB] from \datatable;
			\addlegendentry{\tbb}
			\addplot [color = graph3, mark=triangle*] table [y = M6-TBU] from \datatable;
			\addlegendentry{\itbu}
			\addplot [color = graph4, mark=diamond*] table [y = M6-TBB] from \datatable;
			\addlegendentry{\itbb}
			
			\end{axis}
			\end{tikzpicture}
		}
	\end{minipage}
	\caption{Google File-system (n=10,000, s=5,000, T=2, $\nSim=1000$)}
\end{figure}

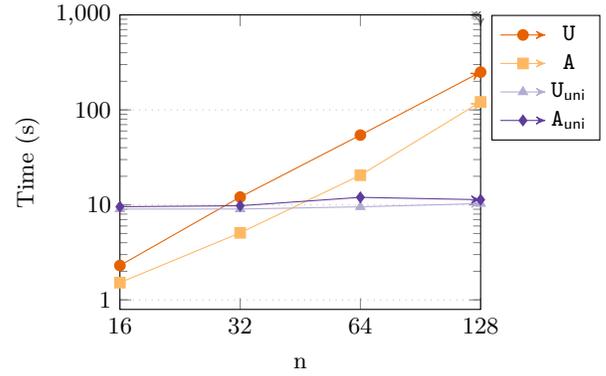
\begin{figure}
		\begin{minipage}{0.45\textwidth}
		\centering
		{\renewcommand{\arraystretch}{1.5}%
							\setlength{\tabcolsep}{9pt}
			\begin{tabular}{rrrrr}
				\toprule
				& & \multicolumn{2}{c}{Explored} \\
				\cmidrule{3-4}
				\thead{n} & \thead{States} & \thead{$\straauni$} & \thead{\%} \\ \midrule
				16 & 10,130 & 96 & 0.95 \\
				32 & 38,674 & 112 & 0.29 \\
				64 & 151,058 & 168 & 0.11 \\
				128 & 597,010 & 274 & 0.05 \\ 
				\bottomrule
			\end{tabular}
		}
	\end{minipage}
	\hfill
	\pgfplotstableread{figures/ftwc.data}
	\datatable
	\begin{minipage}{0.45\textwidth}
		\centering
		\subfloat{
			\begin{tikzpicture}
			\begin{axis}[
			xmode=log,
			ymode=log,
			log ticks with fixed point,
			height=5.5cm,
			xlabel={n},
			ylabel={Time (s)},
			xmin=16, xmax=128,
			ymin=0, ymax=1000,
			xtick={16,32,64,128,256},
			ytick={0, 1, 10, 100, 1000},
			legend pos= outer north east,
			ymajorgrids=true,
			grid style=dotted,
			]
			
			\addplot [color = graph1, mark=*] table [y = TBU] from \datatable;
			\addlegendentry{\tbu}
			\addplot [color = graph2, mark=square*] table [y = TBB] from \datatable;
			\addlegendentry{\tbb}
			\addplot [color = graph3, mark=triangle*] table [y = M6-TBU] from \datatable;
			\addlegendentry{\itbu}
			\addplot [color = graph4, mark=diamond*] table [y = M6-TBB] from \datatable;
			\addlegendentry{\itbb}
			
			\end{axis}
			\end{tikzpicture}
		}
	\end{minipage}
	\caption{Fault Tolerant Workstation Cluster (T=1,000, $\nSim=10,000$)}
\end{figure}

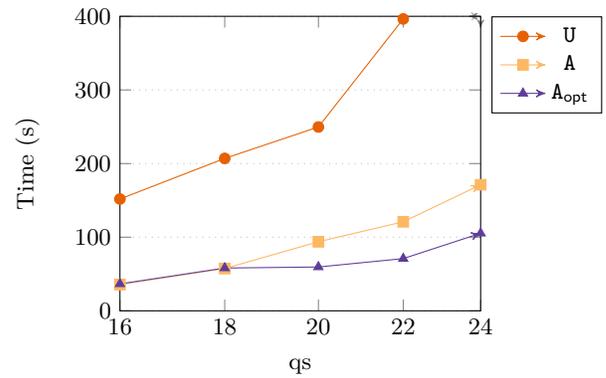
\begin{figure}
		\begin{minipage}{0.45\textwidth}
		\centering
		{\renewcommand{\arraystretch}{1.5}%
							\setlength{\tabcolsep}{9pt}
			\begin{tabular}{rrrrr}
				\toprule
				& & \multicolumn{2}{c}{Explored} \\
				\cmidrule{3-4}
				\thead{qs} & \thead{States} & \thead{$\straaopt$} & \thead{\%} \\ \midrule
				16 & 1,591,817 & 18,714 & 1.18 \\
				18 & 2,496,681 & 18,915 & 0.76 \\
				20 & 3,741,449 & 21,088 & 0.56 \\
				22 & 5,402,153 & 21,361 & 0.40 \\
				24 & 7,562,505 & 23,309 & 0.31 \\ 
				\bottomrule
			\end{tabular}
		}
	\end{minipage}
	\hfill
	\pgfplotstableread{figures/polling.data}
	\datatable
	\begin{minipage}{0.45\textwidth}
		\centering
		\subfloat{
			\begin{tikzpicture}
			\begin{axis}[
			xmode=log,
			log ticks with fixed point,
			height=5.5cm,
			xlabel={qs},
			ylabel={Time (s)},
			xmin=16, xmax=24,
			ymin=0, ymax=400,
			xtick={16,18, 20, 22, 24},
			ytick={0, 100, 200, 300, 400},
			legend pos= outer north east,
			ymajorgrids=true,
			grid style=dotted,
			]

			\addplot [color = graph1, mark=*] table [y = TBU] from \datatable;
			\addlegendentry{\tbu}
			\addplot [color = graph2, mark=square*] table [y = TBB] from \datatable;
			\addlegendentry{\tbb}
			\addplot [color = graph4, mark=triangle*] table [y =  M6-TBB-FEEDBACK] from \datatable;
			\addlegendentry{\itbbplus}
			
			\end{axis}
			\end{tikzpicture}
		}
	\end{minipage}
	\caption{Polling System (jt=4, g=\one, T=2, $\nSim=60,000$)}
\end{figure}

\end{document}